\documentclass[10pt]{article}
\usepackage[margin=1in]{geometry}

\usepackage{hyperref}       % hyperlinks
\usepackage{url}            % simple URL typesetting
\usepackage{graphicx}
\usepackage{booktabs} % for professional tables
\usepackage{amsmath}
\usepackage{amssymb}
\usepackage{amsthm}
\usepackage{graphicx}
\usepackage{xcolor}
\usepackage{enumitem}
\usepackage[numbers]{natbib}
\usepackage{verbatim}
\usepackage{algorithm}
\usepackage{algorithmic}
\usepackage{multirow}
\usepackage{caption}
\usepackage{subcaption}
\usepackage[bottom]{footmisc}
\usepackage{listings}
\usepackage{authblk}

\newtheorem{lemma}{Lemma}
\newtheorem{theorem}{Theorem}
\newtheorem{corollary}{Corollary}

\newtheorem{proposition}{Proposition}

\title{Federated Heavy Hitters Discovery with Differential Privacy}

\author[1]{Wennan Zhu\thanks{Work done while interning at Google.}}
\author[2]{Peter Kairouz}
\author[2]{Brendan McMahan}
\author[2]{Haicheng Sun}
\author[2]{Wei Li}
\affil[1]{Rensselaer Polytechnic Institute. zhuw5@rpi.edu}
\affil[2]{Google. \{kairouz, mcmahan, haicsun, liweithu \}@google.com }

\date{}
\begin{document}

\maketitle

\begin{abstract}
The discovery of heavy hitters (most frequent items) in user-generated data streams drives improvements in the app and web ecosystems, but can incur substantial privacy risks if not done with care. To address these risks, we propose a distributed and privacy-preserving algorithm for discovering the heavy hitters in a population of user-generated data streams. We leverage the sampling and thresholding properties of our distributed algorithm to prove that it is inherently differentially private, without requiring additional noise. We also examine the trade-off between privacy and utility, and show that our algorithm provides excellent utility while also achieving strong privacy guarantees. A significant advantage of this approach is that it eliminates the need to centralize raw data while also avoiding the significant loss in utility incurred by local differential privacy. We validate our findings both theoretically, using worst-case analyses, and practically, using a Twitter dataset with 1.6M tweets and over 650k users. Finally, we carefully compare our approach to Apple's local differential privacy method for discovering heavy hitters. 
\end{abstract}

\section{Introduction}

Discovering the heavy hitters in a population of user-generated data streams plays an instrumental role in improving mobile and web applications. For example, learning popular out-of-dictionary words can improve the auto-complete feature in a smart keyboard, and discovering frequently-taken actions can provide an improved in-app user experience. Naively, a service provider can learn the popular elements by first collecting user data and then applying state-of-the-art centralized heavy hitters discovery algorithms \citep{Cormode2003, Cormode2008,charikar2002finding}. However, collecting and analyzing data from users can introduce privacy risks. 

To overcome some of these risks, the service provider can use the central model of differential privacy (DP) to provide internal or external analysts with a privacy-preserving set of learned heavy hitters \citep{dwork2006calibrating, dwork2006our, dwork2008differential, dwork2010differential, bhaskar2010discovering, dwork2014algorithmic}. However, this approach requires that users trust the service provider with their raw data. And even with a fully trusted service provider, tighter privacy regulations, such as Europe's General Data Protection Regulation (GDPR), the risk of hacks and other data breaches, and subpoena powers may encourage service providers to collect less data from their users. 

The local model of DP \citep{warner1965randomized, evfimievski2004privacy,kasiviswanathan2008ldp} addresses the above concerns by requiring users to perturb their data locally before sharing it with a service provider. Google \citep{erlingsson2014rappor}, Apple \citep{apple2017}, and others \citep{,ding2017collecting,kenthapadi2018pripearl} have deployed local DP algorithms. However, a large body of fundamental work shows that in the context of learning distributions and heavy hitters, local DP often leads to a significant reduction in utility \citep{kairouz2014extremal, wang2017locally, bassily2017practical, kairouz2016discrete, ye2018optimal,duchi2013local, cormode2018marginal}. As we show (e.g., Table~\ref{table:precision_recall}), there are regimes where local DP is infeasible for practical use. Our goal is to provide practical algorithms that provide more privacy than prior approaches in such regimes, while maintaining sufficient utility (precision and recall).\footnote{Whether or not a given approach provides sufficient privacy for a particular application is largely a domain-dependent policy question beyond the scope of this work; our goal is to expand the set of approaches available.}

Our work builds on recent advances in federated learning (FL) \citep{FL_BLOG,konevcny2016federated, mcmahan2017communication} to bridge the utility gap between the local and central models of DP. Our proposed algorithm retains the essential privacy ingredients of FL: (a) no raw data collection (only ephemeral, focused updates from a random subset of users are sent back to the service provider), (b)  decentralization across a large population of users (most users will contribute only 0 or 1 times), (c) interactivity in building an aggregate understanding of the population. However, unlike existing FL algorithms where the goal is to learn a prediction model, our work introduces a new federated approach that allows a service provider to discover the heavy hitters.

% Differential privacy (DP), a rigorous privacy notion that has been carefully  studied over the last decade \citep{dwork2006calibrating, dwork2006our, dwork2008differential, dwork2014algorithmic} and widely adopted in industry \citep{ding2017collecting,apple2017,kenthapadi2018pripearl,erlingsson2014rappor}, provides the ability to make such strong formal privacy guarantees.

%This is particularly important since the learned information (model or set of frequent sequences) is shared with potentially compromised users. 
%Indeed, with sufficient collusion power, an adversary can, in some instances, learn sensitive information about participating users. 

\paragraph{Contributions} We develop an interactive heavy hitters discovery algorithm that achieves central DP while minimizing the data collected from users. In contrast to classical frequency estimation problems, our goal is to discover the heavy hitters but not their frequencies\footnote{Observe that once the popular items are discovered, learning their frequencies can be done using off-the-shelf DP techniques.}. For example, in a smart mobile keyboard application, our algorithm allows a service provider to discover out-of-dictionary words and add them to the keyboard's dictionary, allowing these words to be automatically spell-corrected and typed using gesture typing.  

We assume, without loss of generality,\footnote{Regardless of the items' data type, they can always be represented by a sequence of bits.} that items (e.g., words) in user-generated data streams have a sequential structure (e.g., sequence of characters). Thus, we refer to items as sequences and leverage their sequential structure to build our algorithm. Our algorithm is interactive and runs in multiple rounds. In each round, a randomly selected set of users transmit a ``vote'' for a one element extension to popular prefixes discovered in previous rounds. The server then aggregates the received votes using a trie data structure, prunes nodes that have counts that fall below a chosen threshold $\theta$, and continues to the next round.

We prove that our algorithm is inherently differentially private, and show how the parameters of the algorithm can be chosen to obtain precise privacy guarantees (see Theorem \ref{thm-dp} and Corollary \ref{coro-theta-lambert}). When the number of users $n \ge 10^4$ and the sequences have a length of at most 10, our algorithm guarantees $(2, \frac{1}{ n^2})$-differential privacy while achieving good utility (see Figure \ref{fig-fix-prob}). See Table \ref{table:eps-delta} for the DP parameters we can provide for various population sizes.

A key property of our algorithm is that it is sufficient for the service provider to receive only the set of extensions to the trie with votes that exceed a threshold $\theta$, and the set of possible extensions is finite and known at the start of each round. A simple implementation of our algorithm would have the service provider directly receive each selected user's anonymous vote, and then immediately aggregate and threshold these votes in memory, with no persistence of the unaggregated votes.

However, our algorithm was explicitly designed to allow it to be implemented using aggregation schemes that further limit the information the service provider receives. In particular, a cryptographic secure sum protocol such as that of \citep{bonawitz2016practical} can be used to count votes, so the service provider never sees individual votes, only the aggregate sum over all users in the round (and only if a sufficient number of users participate). The service provider then is only trusted to apply the threshold $\theta$. An intriguing open question is whether an efficient secure multi-party computation can be developed which also performs the thresholding. Another approach is to use the ESA architecture of \citep{bittau17prochlo} to ensure shuffling and anonymization of the votes.

We have already discussed the privacy advantages of our approach compared to centralized approaches with DP that collect and store raw user data; undoubtedly such approaches could offer even higher utility, but we do not empirically assess this, as it is enough to show our algorithm achieves sufficient utility to be practical in many settings. Rather, we focus our empirical evaluation of utility on a comparison to local DP (in particular \citep{apple2017}), demonstrating that our algorithm obtains a strong central DP guarantee and high utility in settings where local DP performs poorly (see Table \ref{table:precision_recall} for details). We use the Sentiment140 dataset, a Twitter dataset with 1.6M tweets and over 650k users \cite{sentiment140}. For Sentiment140, the top 200 words are recalled at a rate close to 1 with $\varepsilon = 4$ and $\delta < 5 \times 10^{-9}$.
% Our worst-case theoretical and experimental analyses show that our algorithm is successful at discovering popular sequences reliably even with a conservative privacy budget. For example, on the    

\paragraph{Related work} 
Federated learning (FL) \citep{mcmahan2017communication,konevcny2016federated,bonawitz19sysml} is a collaborative learning approach that enables a service provider to learn a prediction model without collecting user data (i.e., while keeping the training data on user devices). The training phase of FL is interactive and executes in multiple rounds. In each round, a randomly chosen small set of online users download the latest model and improve it locally using their training data. Only the updates are then sent back to the service provider where they are aggregated and used to update the global model. Much of the existing works are in the context of learning prediction models. Our work differs in that it focuses on federated algorithms for the discovery of heavy hitters.

Differential privacy (DP) is a rigorous privacy notion that has been carefully studied over the last decade \citep{dwork2006calibrating, dwork2006our, dwork2008differential, dwork2014algorithmic} and widely adopted in industry \citep{ding2017collecting,apple2017,kenthapadi2018pripearl,erlingsson2014rappor}. It provides the ability to make strong formal privacy guarantees by bounding the worst-case information loss. There is a rich body of work on distribution learning, frequent sequence mining, and heavy-hitter discovery both in the central and local models of DP \citep{bhaskar2010discovering, bonomi2013mining, diakonikolas2015differentially, xu2016differentially, zhou2018frequent, kairouz2016discrete, wang2017locally, bassily2017practical,  acharya2018communication, ye2018optimal, avent2017blender, Bun2018, cormode2018marginal}, and some recent works combine FL with central DP \citep{geyer2017differentially, mcmahan2018learning}. The central model of DP assumes that users trust the service provider with their raw data while the local one gets away with this assumption. Thus, the utility loss is not as severe in the central model where the service provider may have access to the entire dataset. Our work bridges these existing models of privacy in that it allows an honest-but-curious service provider to learn the popular sequences in a centrally differentially private way, while only having access to minimal data: a randomly chosen user submits one character extension to an already discovered popular prefix.

Methods that provide DP typically involve adding noise, such as Gaussian noise, to the data before releasing it. In this work, we show that DP can be obtained without the addition of any noise by relying exclusively on random sampling and trie pruning which achieves $k$-anonymity. The connection between DP, random sampling, and $k$-anonymity has previously appeared in the literature \citep{chaudhuri2006random, li2012sampling, gehrke2012crowd}. However, our approach and analysis are different in two fundamental ways. 
First, existing methods show how sampling and enforcing $k$-anonymity at the sequence level (in a centralized setting) can achieve central DP.  When applied to our decentralized setting, such approaches have the disadvantage of revealing the entire sequences held by sampled users. On the contrary, our approach explores how interactivity, random sampling, and $k$-anonymity can achieve central DP while also drastically minimizing the data a user shares with the service provider. 
Second, our sampling method is different from existing methods that sample records from a centralized database in an i.i.d fashion (referred to as \textit{Poisson sampling}). Under Poisson sampling, the number of chosen users can vary drastically across rounds, making such approach incompatible with existing federated learning production systems such as \citep{bonawitz19sysml}. Instead, we sample (uniformly at random) a fixed number of users in each round. Combined with interactivity over rounds, this different sampling strategy makes our approach and proof techniques different from existing ones.

Our trie-based heavy hitters (TrieHH) algorithm exploits the hierarchical structure of user-generated data streams to interactively maintain a trie structure that contains the frequent sequences. The idea of using trie-like structures for finding frequent sequences in data streams has been explored before in \citep{Cormode2003, bassily2017practical}. However, the work of \citet{Cormode2003} predates differential privacy and the TreeHist algorithm of \citet{bassily2017practical} is non-interactive, relies on sketching, achieves local DP using the randomized response, and assumes the existence of public randomness. Our approach is interactive in nature, does not use sketching or offer local DP, and does not require public randomness. The only similarity between these two approaches is the use of a trie-like data structure that maintains a list of popular prefixes, a practice that is common for efficient discovery of heavy hitters (even under no privacy constraints). In fact, the differences between these two approaches lead to a fundamentally different privacy-utility trade-off and make private heavy-hitter discovery feasible even for small-to-moderate populations. 

In Section \ref{sec:experiments}, we compare TrieHH with Apple's Sequence Fragment Puzzle (SFP) algorithm, a state-of-the-art sketching based algorithm for discovering heavy hitters with local DP \citep{apple2017}. Similar to TreeHist, SFP is also a count sketch based algorithm. However, instead of pruning by a tree structure, SFP estimates high frequency substring fragments and then stitches them together to get full length heavy hitters. We provide our source code implementation of SFP at \url{https://github.com/tensorflow/federated/tree/master/tensorflow_federated/python/research/triehh}, and a detailed description of this algorithm in Section \ref{sec:sfp} of the appendix.

%\section{Preliminaries and Organization}
\section{Preliminaries}
\label{sec:model}

\paragraph{Model and notation} We consider a population of $n$ users $\mathcal{D}=\{ u_1, u_2, \dots, u_n \}$, where user $i$ has a collection of items $\{w_{i1}, w_{i2}, \cdots, w_{iq}\}$. We abuse notation and use $\mathcal{D}$ to refer to both the set of all users and set of all items. Without loss of generality, we assume that the items have a sequential structure and refer to them as sequences. More precisely, we express an item $w$ as a sequence $w = c_1 c_2 \dots c_{|w|}$ of $|w|$ elements. For example, in our experiments (see Section \ref{sec:experiments}), we focus on discovering heavy-hitter words in a population of tweets generated by Twitter users. Therefore, each user has a collection of words, and each word can be expressed as a sequence of ASCII characters. We assume that the length of any sequence is at most $L$.

For any set $\mathcal{D}$, we build a trie via a randomized algorithm $\mathcal{M}$ to obtain an estimate of the heavy hitters. We let $p_i(w)$ denote the prefix of $w$ of length $i$. For a trie $T$ and a prefix $p = c_1, c_2 \dots c_i$, we say that $p \in T$ if there exists a path $(\text{root}, c_1, c_2, \dots, c_i)$ in $T$. Also, let $T_i$ denote the subtree of $T$ that contains all nodes and edges from the first $i$ levels of $T$. Suppose $(root, c_1, c_2, \dots, c_i)$ is a path of length $i$ in $T_i$. Growing the trie from $T_i$ to $T_{i+1}$ by ``adding prefix $(root, c_1, c_2, \dots, c_i, c_{i+1})$ to $T_{i}$'' means appending a child node $c_{i+1}$ to $c_i$.

\paragraph{Differential privacy}  A randomized algorithm $\mathcal{M}$ is $(\varepsilon, \delta)$-differentially private iff for all $\mathcal{S} \subseteq Range(\mathcal{M})$, and for all adjacent datasets $\mathcal{D}$ and $\mathcal{D}'$:
\begin{equation}
\label{eq-dp}
P(\mathcal{M}(\mathcal{D}) \in \mathcal{S}) \le e^\varepsilon P(\mathcal{M}(\mathcal{D}') \in \mathcal{S}) + \delta.
\end{equation}
We adopt user-level adjacency where $\mathcal{D}$ and $\mathcal{D'}$ are adjacent if $\mathcal{D'}$ can be obtained by adding all the items associated with a single user from $\mathcal{D}$ \citep{mcmahan2018learning}. This is stronger than the typically used notion of adjacency where $\mathcal{D}$ and $\mathcal{D'}$ differ by only one item \citep{dwork2014algorithmic}.

\paragraph{Paper organization} We focus in Section \ref{sec:single_sequence} on the setting where each user has a single sequence ($q=1$). We present the basic version of our algorithm, prove that it is differentially private, and provide worst-case utility guarantees. Combining key insights from Section \ref{sec:single_sequence}, we handle the more general case of multiple sequences per user in Section \ref{sec:multiple_sequences}. We present, in Section \ref{sec:experiments}, extensive simulation results on the Sentiment140 Twitter dataset of 1.6M tweets \citep{sentiment140}. We conclude our paper with a few interesting and non-trivial extensions in Section \ref{sec:conclusion}. All proofs and additional experiments are deferred to the accompanying supplementary material.

\section{Single Sequence per User}
\label{sec:single_sequence}

In this section, we consider a simple setting where each user has single sequence. Much of the intuition behind the algorithm and privacy guarantees we present in this section carry over to the more realistic setting of multiple sequences per user. 

% Our algorithm runs in multiple rounds.  In the $i^{th}$ round, randomly selected users receive a trie containing the popular prefixes (up to length $i$) that have been learned so far. If a user’s word has a length $i$ prefix that is in the trie, they declare the length $i+1$ prefix of the word they have. Otherwise, they do nothing.

% In each round, a subset of users are randomly selected and asked to vote for a prefix of their sequences. Then the server aggregates the results in a trie structure as the global model, with prefixes that have votes exceeding a fixed threshold.

% In each round, the service provider selects $m = 10$ random users, asks them vote for prefixes of their words, and stores the prefixes that receive votes greater than or equal to a threshold $\theta = 2$ in a trie.

We describe the proposed approach via a simple example (shown in Figure \ref{fig-flow}) where the goal is to discover popular words. Suppose we have $n = 20$ users and each user has a single word. Assume there are three popular words: ``star'' (on 3 devices), ``sun'' (on 4 devices) and ``moon''(on 4 devices). The rest of the words appear once each. We add a ``\$'' to the end of each word as an ``end of sequence'' (EOS) symbol. In each round, the service provider selects $m = 10$ random users, asks them to vote for a prefix of their word (as long as it is an extension of the prefixes learned in previous rounds), and stores the prefixes that receive votes greater than or equal to $\theta = 2$ in a trie. In the example in the figure, two prefixes ``s'' and ``m'' of length 1 grow on the trie after the first round. This means that among the 10 randomly selected users, at least two of them voted for ``s'' and at least another two voted for ``m''. Observe that users who have ``sun'' and ``star'' share the first character ``s'', so ``s'' has a significant chance of being added to the trie. In the second round, 10 users are randomly selected and provided with the depth 1 trie learned so far (containing ``s'' and ''m''). In this round, a selected user votes for the length 2 prefix of their word only if it starts with an ``s'' or ``m''. The service provider then aggregates the received votes and adds a prefix to the trie if it receives at least $\theta = 2$ votes. In this particular example, prefixes ``st'', ``su'', and ``mo'' are learned after the second round. This process is repeated for prefixes of length 3 and 4 in the third and the fourth rounds, respectively. After the fourth round, the word ``sun\$'' is completely learned, but the prefix ``sta'' stopped growing. This is because at least two of the three users holding ``star'' were selected in the second and third round, but less than two were chosen in the fourth one. The word ``moon\$'' is completely learned in the fifth round. Finally, the algorithm terminates in the sixth round, and the completely learned words are ``sun\$'' and ``moon\$''.

\begin{minipage}{.45\textwidth}
        \centerline{\includegraphics[width=0.9\textwidth]{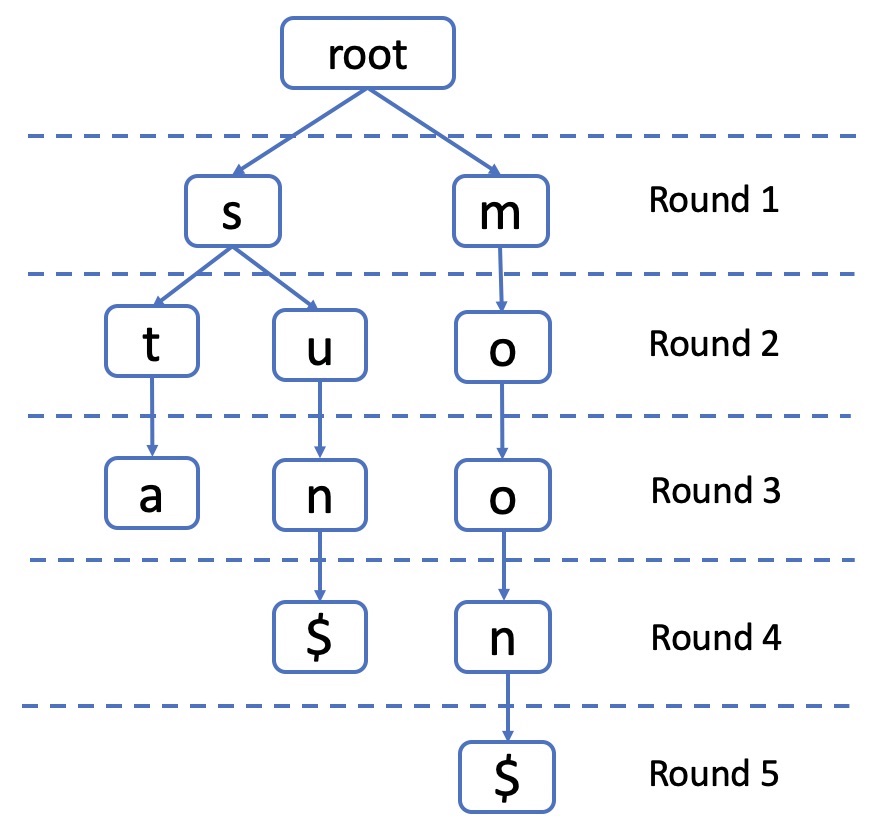}}
      \captionof{figure}{Example run of Algorithm \ref{trie-alg}.}
      \label{fig-flow}
\end{minipage}\hspace{0.1cm}\noindent\fbox{\begin{minipage}{.48\textwidth}
    \begin{algorithmic}
      \STATE {\bfseries Input:}  A set $\mathcal{D} = \{ u_1, u_2, \dots, u_n \}$ that have words $\{ w_1, w_2, \dots, w_n \}$. A threshold $\theta$. Batch size $m = \gamma \sqrt{n}$.\\
      \STATE {\bfseries Output:} A trie $T$.
      \STATE Set $T = \{ root \}$; $T_{old} = None$; $i$ = 1; \\
      \WHILE {$T ~ != ~T_{old}$}
      \STATE Choose $m$ users from $\mathcal{D}$ randomly to get a set $\tilde{\mathcal{X}}$ of sequences;
      \STATE $T_{old} = T$;
      \STATE $T$ = $\mathcal{V}(\tilde{\mathcal{X}}, T, \theta, i)$;  \ \ $i$++;
      \ENDWHILE
      \STATE return $T$;\\
    \end{algorithmic}
    \captionof{algorithm}{Trie-based Heavy Hitters $\mathcal{M}(\mathcal{D}, \theta, \gamma)$}
      \label{trie-alg}
\end{minipage}}

% \begin{figure}[htb]
%     \vspace{.1in}
%     \centerline{\fbox{\includegraphics[width=0.4\textwidth]{flow_short.jpg}}}
%     \vspace{.1in}
%   \caption{Example run of Algorithm \ref{trie-alg}.}
%   \label{fig-flow}
% \end{figure}

% \begin{algorithm}[htb]
%   \caption{Trie-based Heavy Hitters $\mathcal{M}(\mathcal{D}, \theta, \gamma)$}
%   \label{trie-alg}
%   \begin{algorithmic}
%       \STATE {\bfseries Input:}  A set $\mathcal{D} = \{ u_1, u_2, \dots, u_n \}$ that have words $\{ w_1, w_2, \dots, w_n \}$. A threshold $\theta$. Batch size $m = \gamma \sqrt{n}$.\\
%       \STATE {\bfseries Output:} A trie $T$.
%       \STATE Set $T = \{ root \}$; $T_{old} = None$; $i$ = 1; \\
%       \WHILE {$T ~ != ~T_{old}$}
%       \STATE Choose $m$ users from $\mathcal{D}$ randomly to get a set $\tilde{\mathcal{X}}$ of sequences;
%       \STATE $T_{old} = T$;
%       \STATE $T$ = $\mathcal{V}(\tilde{\mathcal{X}}, T, \theta, i)$;  \ \ $i$++;
%       \ENDWHILE
%       \STATE return $T$;\\
%     \end{algorithmic}
% \end{algorithm}

\begin{algorithm}[htb]
   \caption{Algorithm $\mathcal{V}(\tilde{\mathcal{X}}, T_{in}, \theta, i)$ to grow a trie by one level with a set of sequences.}
   \label{vote-alg}
\begin{algorithmic}
   \STATE {\bfseries Input:}
         A set of sequences $\tilde{\mathcal{X}} = \{ w_1', w_2', \dots, w_m' \}$. An input trie $T_{in}$ with i levels. A threshold $\theta$.
   \STATE {\bfseries Output:} An output trie.

   \STATE Initialize Candidates[$w_j'$] = 0 for all $w_j' \in \tilde{\mathcal{X}}$;
   \FOR{each sequence $w_j'$ in $\tilde{\mathcal{X}}$ that $|w_j'| \ge i$ and $p_{i-1}(w_j') \in T_{in}$}
     \STATE Candidates[$p_i(w_j')$]++;
   \ENDFOR
%   \IF{$\exists$ prefix p, that Candidates[p] $\ge \theta$}
%   \STATE Add all such prefixes to $T_{in}$; \ \ $T_{out}$ = $T_{in}$;
%   \ENDIF
   \STATE  return $T_{in} + \{p \mid \text{Candidates}[p] \ge \theta\}$;\\
\end{algorithmic}
\end{algorithm}

To describe the algorithm formally, for a set of users $\mathcal{D}$, our algorithm $\mathcal{M}(\mathcal{D}, \theta, \gamma)$ runs in multiple rounds, and returns a trie that contains the popular sequences in $\mathcal{D}$. In each round of the algorithm, a batch of size $m = \gamma \sqrt{n}$ (with $\gamma \ge 1$) users are selected uniformly at random from $\mathcal{D}$. Note that there are interesting trade-offs between the utility and privacy with different choices of $\gamma$, which we will discuss later.

%Note that $m$ could be chosen differently, and there are interesting trade-offs between the utility and privacy as a function of $m$. We will discuss these trade-offs later in this section.

In the $i^{th}$ round, randomly selected users receive a trie containing the popular prefixes that have been learned so far. If a user's sequence has a length $i -1$ prefix that is in the trie, they declare the length $i$ prefix of the sequence they have. Otherwise, they do nothing. Prefixes that are declared by at least $\theta \approx \log n$ selected users grow on the $i^{th}$ level of the trie. Note that we grow at most one level of the trie in each round of the algorithm. Thus, if $c_1, \dots, c_{i-1} \notin T_{i-1}$, then $c_1, \dots, c_{i-1}, c_i$ cannot be in $T_i$. The final output of $\mathcal{M}$ is the trie returned by the algorithm when it stops growing. Algorithm \ref{trie-alg} describes our distributed algorithm and Algorithm \ref{vote-alg} shows a single round of the algorithm to grow one level of the trie.

Given the final trie, we extract the heavy-hitter sequences learned by Algorithm \ref{trie-alg} by simply outputting the discovered prefixes from the root to leaves that end with \$ (the EOS symbol). Note that the non-EOS leaves also represent frequent prefixes in the population, which might still be valuable depending on the application.

%We use $\mathcal{M}_i(\mathcal{X}, \theta, \gamma, T_{i-1})$ to represent the intermediate trie $T_i$ given $T_{i-1}$ in the $i^{th}$ round of the algorithm. Let $T_0$ denote the root node.

%In Theorem \ref{thm-dp}, we will show that when $n \ge 10^4$, choosing $\theta = \lceil \log_{10} n + 6 \rceil$, $\gamma \ge 1$, $\gamma + \frac{1}{\gamma} \le 0.874 \sqrt{n}$ and $\gamma \le \frac{\sqrt{n}}{\theta + 1}$, ensures that Algorithm \ref{trie-alg} is $\left(L \ln \left(1 + \frac{1}{\frac{\sqrt{n}}{\gamma \theta} - 1}\right), \text{ } \frac{1}{100n}\right)$-differentially private.

\paragraph{Privacy guarantees} Algorithm \ref{trie-alg} has several privacy advantages: (a) randomly chosen users vote on a single character extension to an already discovered popular prefix, (b) the votes are ephemeral (i.e., never stored), and (c) a total of $L\gamma\sqrt{n}$ randomly chosen users participate in the algorithm. More importantly, sequences discovered by Algorithm \ref{trie-alg} are $k$-anonymous with $k = \theta$, and as shown in the theorem below, the output of Algorithm \ref{trie-alg} is inherently $(\varepsilon, \delta)$-differentially private -- without the need for additional randomization or noise addition. 

\begin{theorem}
\label{thm-dp}
When $4 \le \theta \le \sqrt{n}$ and $1 \le \gamma \le \frac{\sqrt{n}}{\theta + 1}$, Algorithm \ref{trie-alg} is $(L\ln(1 + \frac{1}{\frac{\sqrt{n}}{\gamma \theta} - 1}), \text{ } \frac{\theta - 2}{(\theta - 3) \theta !})$-differentially private.
\end{theorem}

\begin{proofsketch}
Suppose $\mathcal{D}$ is obtained by adding $w$ to a neighboring $\mathcal{D}'$ and assume $|w| = l$. We first decompose any $\mathcal{S} \subseteq \text{Range}(\mathcal{M})$ into $\mathcal{S}_0 \cup \mathcal{S}_1 \cup \dots \mathcal{S}_l$, where $\mathcal{S}_0 = \{ T \in \mathcal{S} | p_i(w) \notin T, \text{ for } i = 1, 2, \dots, l \}$ and $\mathcal{S}_i = \{ T \in \mathcal{S} | p_1(w), \dots, p_i(w) \in T \text{ and } p_{i+1}, \dots, p_l \notin T \}$ for $i = 1, 2, \dots, l$. Assume there are $k$ users in $\mathcal{D}'$ that have prefix $p_i(w)$. Then we show that when $k$ is large, the ratio between $P(\mathcal{M}(\mathcal{D}) \in \mathcal{S}_i)$ and $P(\mathcal{M}(\mathcal{D}') \in \mathcal{S}_i)$ is small so it could be bounded by $e^\varepsilon$. When $k$ is small,  $P(\mathcal{M}(\mathcal{D}) \in \mathcal{S}_i)$ is small enough so it could be bounded by $\delta$. Intuitively, when $k$ is large, it means prefix $p_i(w)$ is already popular in $\mathcal{D}'$, so the fact that $\mathcal{D}$ has one more user with this prefix does not affect the probability of it showing in the result too much. When $k$ is small, the chance of prefix $p_i(w)$ showing up in the result is very small, even with an extra user with it in $\mathcal{D}$.
\end{proofsketch}

The above result holds for a wide array of algorithm parameters ($L$, $\gamma$, and $\theta$). The following corollary shows how precise privacy guarantees can be obtained by tuning the algorithm's parameters.  

\begin{corollary}
\label{coro-theta-lambert}
To achieve $(\varepsilon, \delta)$-differential privacy, set  $\gamma = (e ^ {\frac{\varepsilon}{L}} - 1)\sqrt{n}/({\theta e ^ {\frac{\varepsilon}{L}}})$ and $\theta = \text{max} \{ 10, \lceil e^{W(C_\delta) + 1} - \frac{1}{2} \rceil, \lceil e^{\frac{\varepsilon}{L}} - 1 \rceil \}$, where $W$ is the Lambert $W$ function \citep{corless1996lambertw} and $C_\delta = e^{-1} \ln(\frac{8}{7\sqrt{2\pi}} \delta^{-1})$. 
Further, when $n \ge 10^4$, choosing $\theta = \lceil \log_{10} n + 6 \rceil$ ensures that Algorithm \ref{trie-alg} is $(\varepsilon,  \text{ } \frac{1}{300n})$-differentially private \footnote{In general, to get a $\delta \le \frac{1}{n^a}$, by standard approximation of the Lambert function, we can choose $\theta \approx a({\ln n}/{\ln \ln n})$.}.
\end{corollary}
%\footnote{This holds as long as $4 \le \theta \le \sqrt{n}$ and $1 \le \gamma \le \frac{\sqrt{n}}{\theta + 1}$}
Table \ref{table:eps-delta} shows how we can choose $\gamma$ and $\theta$ to achieve $(\varepsilon, 1/({300n}))$ and $(\varepsilon, {1}/{n^2})$ for various values of $n$. Since under Algorithm \ref{trie-alg} the privacy loss can be large with probability $\delta$ (unlike mechanisms that rely on explicit noise addition), we focus (almost exclusively) on $\delta < 1/n^2$ in Section \ref{sec:experiments} where we conduct experiments on real data and compare to local differential privacy. 

\begin{figure}[!tbp]
\centering
\begin{minipage}{0.5\textwidth}
        \includegraphics[width=\textwidth]{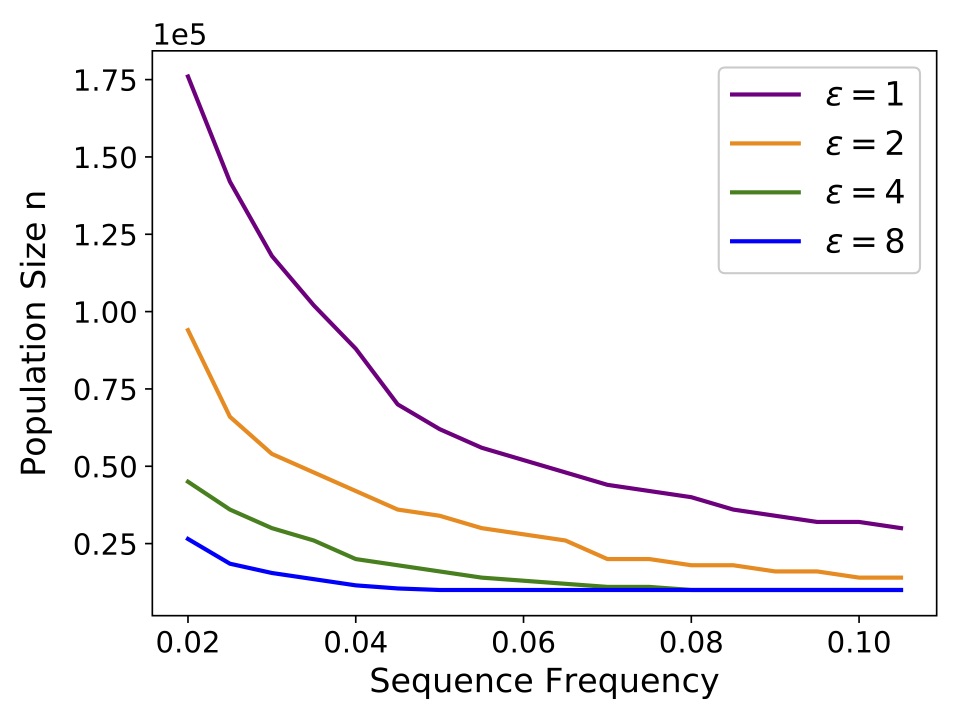}
        \caption{Minimum $n$ required to ensure (via Proposition~\ref{prop:discovery_rate}) a worst-case discovery rate greater than 0.9 for $L = 10$ and $\delta = 1/n^2$.}
        \label{fig-fix-prob}
\end{minipage}\hspace{0.3cm}
\begin{minipage}{0.45\textwidth}
\centering
\captionsetup{type=table} %% tell latex to change to table
  \centering
   \caption{Choices of $\theta$ and $\gamma$ to achieve $\varepsilon = 2$ in two cases: $\delta \le \frac{1}{300n}$ and $\delta \le \frac{1}{n^2}$.}
   \label{table:eps-delta}
    \label{t3}
  \centering
  \renewcommand{\arraystretch}{1.5}
  \renewcommand{\tabcolsep}{1.6mm}
    \begin{tabular}{ | c | c |  c | c | c | }
      \hline
      \multirow{3}{*}{$n$}  & \multicolumn{4}{|c|}{$L=10$} \\ \cline{2-5}
       & \multicolumn{2}{|c|}{$\delta \le \frac{1}{300n}$} & \multicolumn{2}{|c|}{$\delta \le \frac{1}{n^2}$} \\ \cline{2-5}
       & $\theta$ & $\gamma$ & $\theta$ & $\gamma$ \\ \hline
      $10^4$  & 10 & 1.81 & 12 & 1.51 \\ \hline
      $10^5$ & 11 & 5.21 & 14 & 4.09 \\ \hline
      $10^6$  & 12 & 15.10 & 15 & 12.08 \\ \hline
      $10^7$  & 13 & 44.09 & 17 & 33.71 \\ \hline
    \end{tabular}
    \vspace{0.5cm}
\end{minipage}
\end{figure}

\paragraph{Utility guarantees} By the sampling nature of Algorithm \ref{trie-alg}, sequences that appear more frequently are more likely to be learned. The batch size $m$ and threshold $\theta$ could be tuned to trade off utility for privacy. For a user set of size $n$, smaller $m$ and larger $\theta$ achieve better privacy at the expense of lower utility, and vice versa. 
%We note that all the plots we provide in this section represent worst-case lower bounds on utility that apply regardless of the underlying data statistics. We show, in Section \ref{sec:experiments}, that the performance of Algorithm \ref{trie-alg} is much better on real data.

To quantify utility under Algorithm \ref{trie-alg}, we examine the worst-case discovery rate of a sequence (probability of discovering it) as a function of its frequency in the dataset. In particular, we consider the worst-case discovery rate which captures the probability of discovering a sequence assuming that it shares no prefixes with other sequences in the dataset. In the presence of such common prefixes, the discovery rate will only get better (see Section \ref{sec:experiments} for a comparison between worst-case discovery rates and ones that are achievable on real data). 

\begin{proposition}
\label{prop:discovery_rate}
Suppose a sequence appears $W$ times in a dataset of $n$ users where the longest sequence has length $L$. Then the worst-case discovery rate under Algorithm \ref{trie-alg} is given by
\begin{equation}
\label{eq:discovery_rate}
 \left(\frac{1}{\binom{n}{m}} \sum_{i = \theta}^{\min\{W, m\}} \binom{W}{i} \binom{n-W}{m -i}\right) ^ L.
\end{equation}
\end{proposition}

Using Corollary \ref{coro-theta-lambert} and Proposition \ref{prop:discovery_rate}, we can investigate how large the population should be if we want to discover sequences with high probability for a fixed $\varepsilon$. Figure \ref{fig-fix-prob} shows the relationship between sequence frequency and population size $n$ if we want the worst-case discovery rate to be at least 0.9 for different $\varepsilon$'s. Naturally, in order to be discovered with high probability, lower frequency sequences require larger population size, and vice versa. We also need larger populations for stronger privacy guarantees (smaller $\varepsilon$).

% \begin{figure}[htb]
%     %\vspace{.1in}
%     \centerline{\fbox{\includegraphics[width=0.4\textwidth]{fix_prob.jpg}}}
%     %\vspace{.1in}
%   \caption{Minimum $n$ required to ensure (via Proposition~\ref{prop:discovery_rate}) a worst-case discovery rate greater than 0.9 for $L = 10$ and $\delta = 1/n^2$.}
%   \label{fig-fix-prob}
% \end{figure}

\paragraph{Remarks} A few remarks are in order. First, in a production implementation of Algorithm \ref{trie-alg}, not all users may be online in every round of the protocol. In such a situation, the service provider will sample uniformly at random from available users. Therefore, assuming a strong adversary which knows the number and identities of online users in every round, the privacy guarantees will be determined by the number of online users. Second, Theorem \ref{thm-dp} shows that the range of $\gamma$ is: $[1, \sqrt{n}/(\theta + 1)]$. Thus, $\gamma = 1$ is enough to achieve single digit epsilon, and if users are available, it could be increased up to $\sqrt{n}/(\theta + 1)$ to achieve better utility. More importantly, this paper tackles the regime where $n \sim 10^5 - 10^7$ -- see Table 1 for the choices of $\gamma$ to get maximum utility in this setting. Even the upper bound on $\gamma$ is not on the order of $\sqrt{n}$ (but rather 2 to 3 orders smaller than $\sqrt{n}$). For instance, $\gamma \approx 33$ when $\varepsilon = 2$, $\delta = 1/n^2$ and $n = 10^7$. Third, we study the communication cost of Algorithm \ref{trie-alg} in Section \ref{sec:comm_cost} of the appendix, but it is not the central quantity that this work focuses on.

\section{Multiple Sequences per User}
\label{sec:multiple_sequences}

In this section, we consider the more general setting where each user could have more than one sequence on their device. Suppose the population is a set of $n$ users $ \mathcal{D}=\{ u_1, u_2, \dots, u_n \}$, and each user $u_i$ has a set of sequences $\{w_{i1}, w_{i2}, \dots, w_{iq} \}$. 
%Each sequence has a certain number of appearances on the user's device, and if a sequence does not appear on any user's device, we say that it appears 0 times. 

Let $c_i(w_j)$ denote the number of appearances of $w_j$ on $u_i$'s device. We define the local frequency of $w_j$ on $u_i$'s device as $f_i(w_j) = c_i(w_j) / \sum_j c_i(w_j) $. Note that the sum of all the sequences' local frequencies on $u_i$'s device is 1, i.e. $\sum_{j} f_i(w_j) = 1$. If a sequence $w_j$ has 0 appearance on $u_i$'s device, then $f_i(w_j) = 0$. Similarly, for a certain prefix $p_j$, let $c_i(p_j)$ denote the number of appearances of $p_j$ on $u_i$'s device. Then the frequency of $p_j$ on $u_i$'s device is $f_i(p_j) = c_i(p_j) / \sum_j c_i(p_j) $.

We are now ready to generalize Algorithm \ref{trie-alg} to accommodate multiple sequences per user. In each round of the algorithm, we select a batch of $m$ users from $\mathcal{D}$ uniformly at random. A chosen user $u_i$ randomly selects a sequence $w_j \in u_i$ with probability $f_i(w_j)$, i.e., according to its local frequency. Thus, as in Algorithm \ref{trie-alg}, we still select $m$ sequences from $m$ users in every round. The voting step by these $m$ sequences proceeded in the same way described in Algorithm \ref{vote-alg}. Algorithm \ref{trie-alg-multi} shows the full algorithm.

%Algorithm \ref{trie-alg-multi} (in Section \ref{sec:trie-alg-multi} of attached supplementary materials) shows the full algorithm.

\begin{algorithm}[htb]
  \caption{A Trie-based Frequent Sequence Algorithm $\mathcal{M}(\mathcal{D}, \theta, \gamma)$ for Multiple Sequences per User.}
  \label{trie-alg-multi}
\begin{algorithmic}
  \STATE {\bfseries Input:} A set $\mathcal{D} = \{ u_1, u_2, \dots, u_n \}$, A threshold $\theta$. Batch size $m = \gamma \sqrt{n}$.
  \STATE {\bfseries Output:} A trie.
  \STATE Set $T = \{ root \}$; $T_{old} = None$; $i$ = 1; \\
  \WHILE {$T != T_{old}$}
  \STATE Choose $m$ users from $\mathcal{D}$ uniformly at random, denote as $\tilde{\mathcal{X}}$. Initialize $\tilde{\mathcal{X}} = \{\}$.
  \FOR{For each user $u_i \in \tilde{\mathcal{X}}$}
  \STATE Randomly select a sequence $w_j \in u_i$ with respect to its frequency $f_i(w_j)$ in $u_i$, and add $w_j$ to $\tilde{\mathcal{X}}$.
  \ENDFOR
  \STATE $T_{old} = T$;
  \STATE $T$ = $\mathcal{V}(\tilde{\mathcal{X}}, T, \theta, i)$;  \ \ $i$++;
  \ENDWHILE
  \STATE return $T$;\\
\end{algorithmic}
\end{algorithm}

Interestingly, the differential privacy guarantees we obtained in the single sequence setting also hold in the multiple sequence setting. This is formally stated in Corollary \ref{thm-dp-multi}. To get this conclusion, we first provide the following more general (but intuitive) result.
%showing that if a mechanism $M$ achieves $(\varepsilon, \delta)$ record-level DP on a dataset of size $n$, then for a setting that each user has multiple records, the mechanism that first selects one record from each user, and then applies $M$ on the sampled dataset of size $n$ achieves $(\varepsilon, \delta)$ user-level DP. The proof is deferred to the supplementary material.
%Intuitively, in the single sequence per user setting, the probability that a prefix $p$ is added to the trie is equal to the probability that at least $\theta$ users having this prefix are selected in this round. In the multiple sequences setting, each chosen user $u_i$ selects $p$ by a probability of $f_i(p) \leq 1$. This means that the single sequence setting considered in the previous section is the worst-case setting for the differential privacy guarantees. The rigorous proof is deferred to the supplementary material.
\begin{theorem}
\label{thm-dp-multi-general}
Assume mechanism $M$ achieves $(\varepsilon, \delta)$ record-level\footnote{The difference between record-level and user-level DP is in the way neighboring datasets are defined. Under record-level DP, only a single record is varied when comparing $\mathcal{D}$ to $\mathcal{D'}$.} DP on a dataset of size n. Consider a setting where we have $n$ users and an arbitrary number of records per user. Then the mechanism that first selects 1 record per user (deterministically or randomly)  then applies $M$ to the sampled dataset of size n achieves $(\varepsilon, \delta)$ user-level DP. 
\end{theorem}

\begin{corollary}
\label{thm-dp-multi}
When $4 \le \theta \le \sqrt{n}$ and $1 \le \gamma \le \frac{\sqrt{n}}{\theta + 1}$, Algorithm \ref{trie-alg-multi} is $(L\ln(1 + \frac{1}{\frac{\sqrt{n}}{\gamma \theta} - 1}), \text{ } \frac{\theta-2}{(\theta-3) \theta !})$-differentially private.
\end{corollary}

% \begin{theorem}
% \label{thm-dp-multi}
% When $n \ge 10^4$, choose $\theta = \lceil \log_{10} n + 6 \rceil$, $\gamma \ge 1$, $\gamma + \frac{1}{\gamma} \le 0.874 \sqrt{n}$ and $\gamma \le \frac{\sqrt{n}}{\theta + 1}$. Then algorithm \ref{trie-alg-multi} is $(L\ln(1 + \frac{1}{\frac{\sqrt{n}}{\gamma \theta} - 1}), \text{ } \frac{1}{100n})$-differentially private.
% \end{theorem}

% \begin{theorem}
% \label{thm-dp-multi}
% When $n \ge 10^4$, choose $\theta = \lceil \log_{10} n + 6 \rceil$, $\gamma \ge 1$, $\gamma + \frac{1}{\gamma} \le 0.874 \sqrt{n}$ and $\gamma \le \frac{\sqrt{n}}{\lceil \log_{10} n + 6 \rceil + 1}$. Then algorithm \ref{trie-alg-multi} is $(L\ln(1 + \frac{1}{\frac{\sqrt{n}}{\gamma \lceil \log_{10} n + 6 \rceil} - 1}), \text{ } \frac{1}{100n})$-differentially private.
% \end{theorem}

\section{Experiments}
\label{sec:experiments}

We now showcase the performance of the trie-based heavy hitters (TrieHH) algorithm on real data and compare it to Apple's Sequence Fragment Puzzle (SFP) algorithm, a state-of-the-art sketching based algorithm for discovering heavy hitters with local DP \citep{apple2017}. We provide our source code implementation of both SFP and TrieHH at \url{https://github.com/tensorflow/federated/tree/master/tensorflow_federated/python/research/triehh}, and include a detailed description of SFP in Section \ref{sec:sfp} of the appendix. For a fair comparison between SFP and TrieHH, we ``amplify'' the local $\varepsilon_{local}$ used by SFP to a central $(\varepsilon, \delta)$ used in TrieHH according to Theorem 5.3 of \cite{balle2019privacy}. We also focus exclusively on the discovery stage of SFP and do not account for the count estimation stage. Since the trade-off between precision and recall could be tuned by a parameter $T$ \footnote{The parameters are proxies and do not necessarily represent the actual performance of Apple’s system.} under SFP, we compare TrieHH and SFP using precision, recall, and $F_1$ score. We use Sentiment140, a rich Twitter dataset \citep{sentiment140}, and conduct three sets of experiments (see below for details). We run our experiments many times and report averaged utility metrics with 0.95 confidence intervals. 

\begin{figure}[!bp]
  \centering
  \begin{minipage}[t]{0.48\textwidth}
    \includegraphics[width=\textwidth]{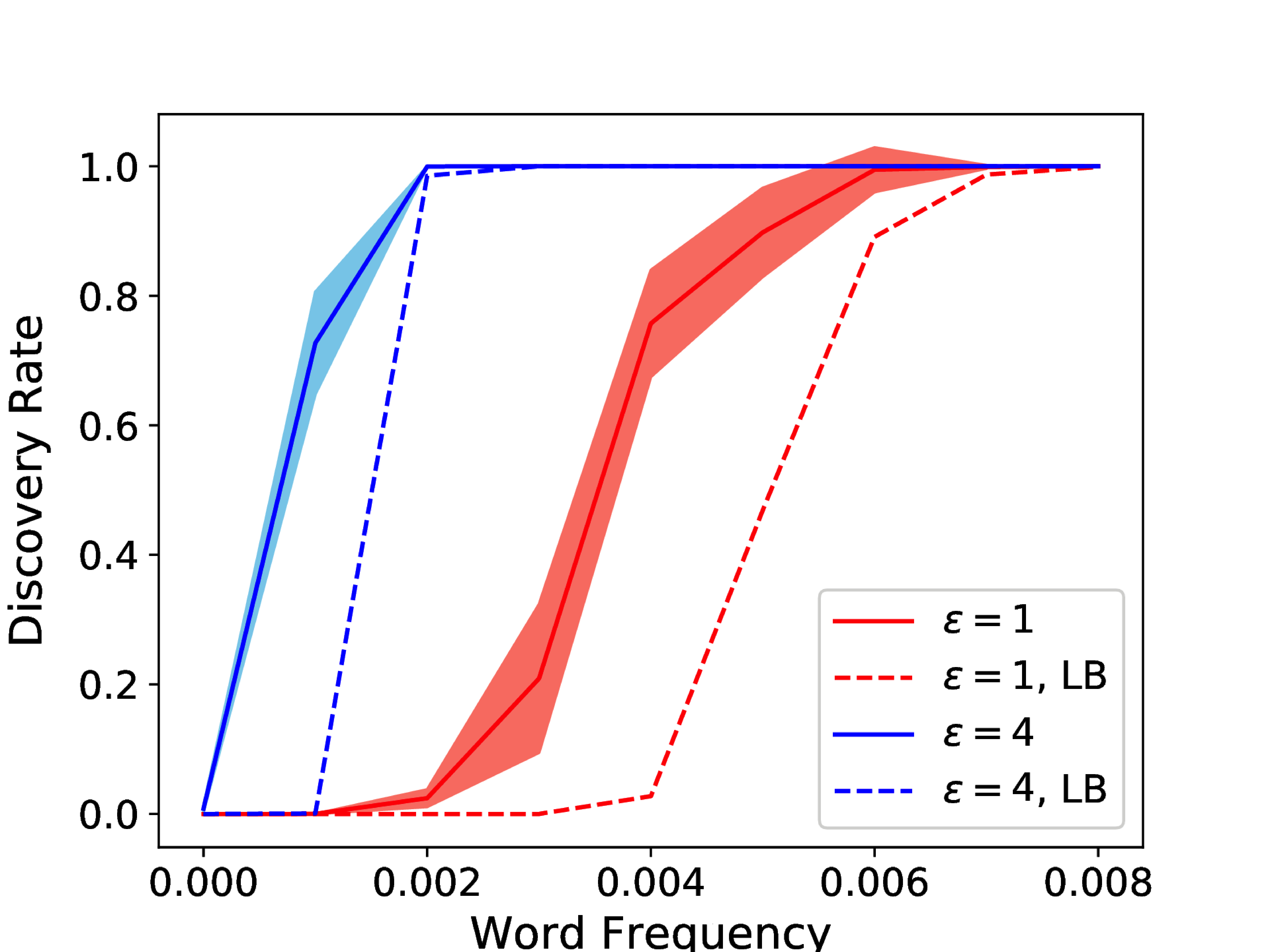}
  \caption{Frequency vs. discovery rate with the theoretical lower bound in the single word setting. ($\delta = 1/n^2$)}
  \label{fig-freq-single}
  \end{minipage}
  \hfill
  \begin{minipage}[t]{0.48\textwidth}
      \includegraphics[width=\textwidth]{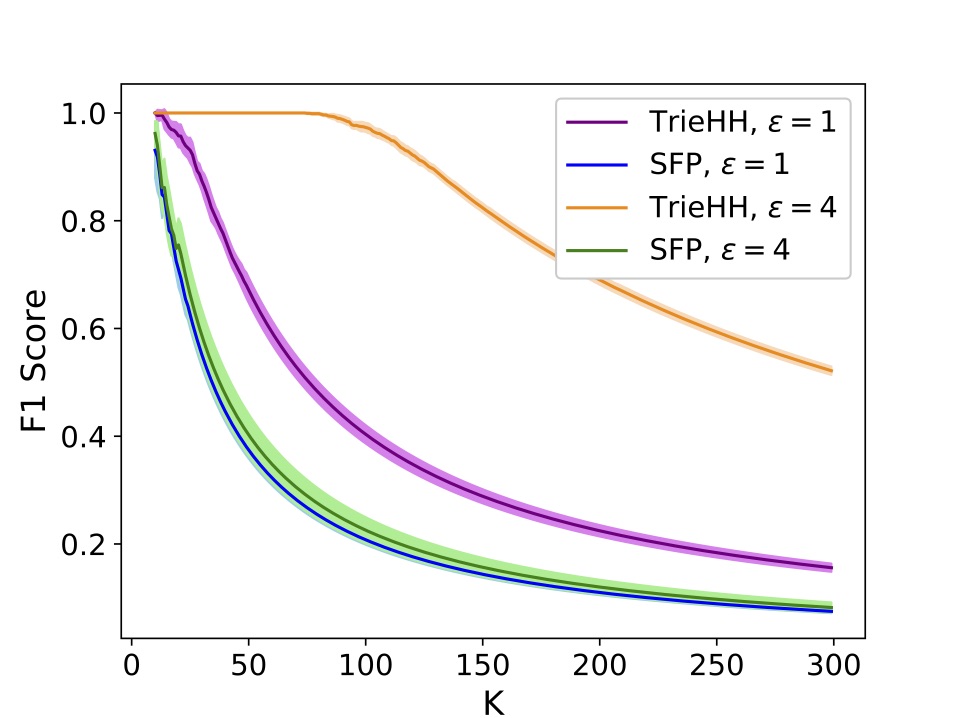}
  \caption{F1 Score of the top K words in the single word setting. $T = 20$ for SFP.}
  \label{fig-recall-single}

  \end{minipage}
\end{figure}

% \begin{figure}[htb]
%     %\vspace{.1in}
%     \centerline{\fbox{  \includegraphics[width=0.43\textwidth]{freq_plot_single.png} }}
%     %\vspace{.1in}
%   \caption{Frequency vs. discovery rate with the theoretical lower bound in the single word setting. ($\delta = 1/n^2$)}
%   \label{fig-freq-single}
% \end{figure}

% \begin{figure}[htb]
%     %\vspace{.1in}
%     \centerline{\fbox{  \includegraphics[width=0.43\textwidth]{f1_single.jpg} }}
%     %\vspace{.1in}
%   \caption{$F_1$ Score of the top K words in the single word setting. $T = 20$ for SFP.}
  
%   \label{fig-recall-single}
% \end{figure}

\paragraph{Single word per user: heavy hitters case}
To simulate this setting that each user has a single word using Sentiment140, we create a dataset by choosing the word with highest local frequency for each user and apply TrieHH on this dataset.
% choose the word with the highest local frequency for each user and assume that the user will only vote for this word if chosen to participate in a round of Algorithm \ref{trie-alg}. In this case, the frequency of a word $w$ in the whole population is the number of users that have $w$ as their highest frequency word divided by the total number of users $n$.
Figure \ref{fig-freq-single} shows the relationship between the word frequencies and the discovery rate using TrieHH. We limit $L$ to 10, set $\delta = 1/n^2$, and choose $\theta$ and $\gamma$ according to Corollary \ref{coro-theta-lambert} to achieve various values of $\varepsilon$. The dashed lines represent the theoretical worst-case bounds on the discovery probability (presented in Section \ref{sec:single_sequence}). Observe that there is a gap between the experimental results and the theoretical worst-case ones. This is because the theoretical bounds assume that sequences share no prefixes with others in the dataset, while in Sentiment140, many English words do share some prefixes. We also study the $F_1$ score of the $K$ highest frequency words in the population. Figure \ref{fig-recall-single} shows the $F_1$ score of the top $K$ words vs. $K$ with comparison to SFP. For SFP, $\varepsilon = 1 \rightarrow  \varepsilon_{local} = 4.29$ and  $\varepsilon = 4 \rightarrow \varepsilon_{local} = 4.96$. Observe that at $\varepsilon = 4$, the top 100 words have an $F_1$ score close to 1 under TrieHH, in comparison to an and $F_1$ score close to 0.2 under SFP.

\paragraph{Single word per user: out-of-vocab (OOV) case} 
To simulate this setting using Sentiment140, OOV words are obtained by first scanning through the dataset and keeping only words that are made up of English letters and a few other symbols (such as "@" and "\#") and then ensuring that these words do not belong to a highly tuned dictionary of over 260k words. After this pre-processing step, the frequencies of the OOV words are calculated and a dataset of size 6M is sampled according to those frequencies. Figure \ref{fig-f1-single-oov} shows the F1 score of the top $K$ words for both TrieHH and SFP. Observe that the curves for both TrieHH and SFP are not monotonically decreasing for small $K$. This is because there are many long words in the top 10 to 20 of the OOV Twitter dataset (corresponding to usernames of trending Twitter users), and both algorithms perform worse for longer words. For larger $K$, the lengths of top words get smaller and more consistent. Table \ref{table:precision_recall} shows recall at $K=50$ and precision for both algorithms with different choices $T$ for SFP. For SFP, $\varepsilon = 1 \rightarrow  \varepsilon_{local} = 5.31$ and  $\varepsilon = 4 \rightarrow \varepsilon_{local} = 5.99$ due to amplification. By increasing $T$ for SFP, there is a gain of recall but the precision also drops dramatically. Some examples of interesting OOV words we have discovered include: "*hugs*", "*sigh*", ":'(", "@tommcfly", "@dddlovato", "\#ff", "\#fb", "b/c", "ya'll". The complete list of heavy-hitter OOV words and discovered ones are given in Section \ref{sec:discovered-oov} of the appendix.

\begin{figure}[!tbp]
\centering
\begin{minipage}{0.48\textwidth}

\includegraphics[width=\textwidth]{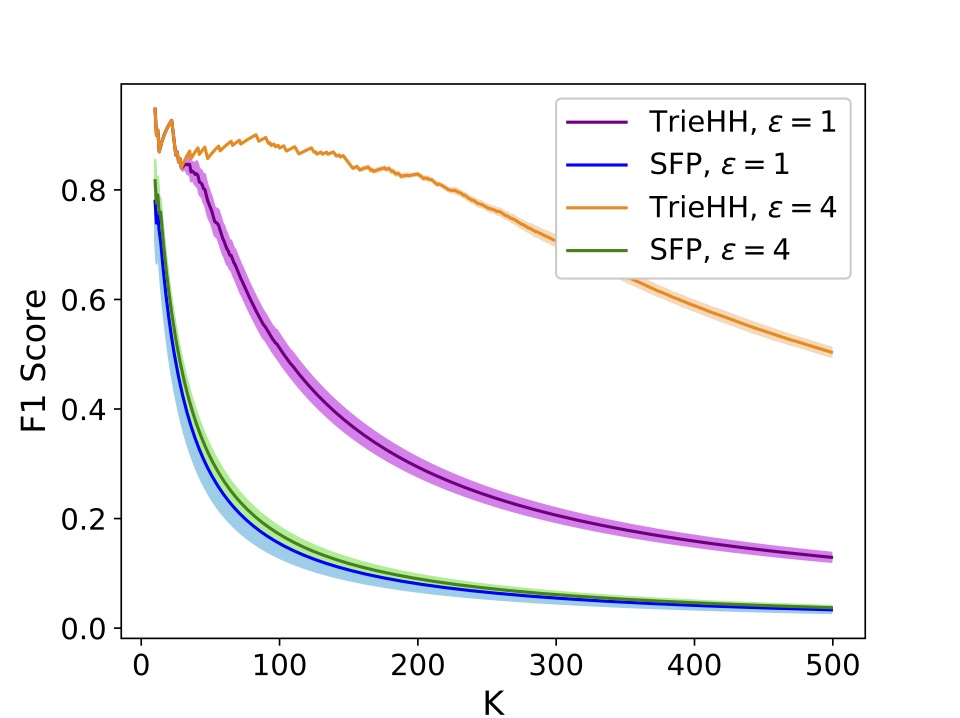}
\caption{F1 Score of the top K words in the single word setting of OOV case ($\delta= 1/n^2$). $T = 20$ for SFP.}
        \label{fig-f1-single-oov}
\end{minipage}\hspace{0.3cm}
\begin{minipage}{0.48\textwidth}
\centering
\captionsetup{type=table} %% tell latex to change to table
  \centering
   \caption{Comparison of recall at $K=50$ and precision between TrieHH and SFP in the OOV setting for $\delta = \frac{1}{n^2}$ and  $T = 20, 80$ under SFP.}
    \label{table:precision_recall}
  \centering
  \renewcommand{\arraystretch}{1.5}
  \renewcommand{\tabcolsep}{1.6mm}
    \begin{tabular}{ | c | c |  c | c | c | }
      \hline
      \multirow{2}{*}{}  
       & \multicolumn{2}{|c|}{$\varepsilon = 1$} & \multicolumn{2}{|c|}{$\varepsilon = 4$} \\ \cline{2-5}
       & Recall & Prec & Recall & Prec \\ \hline
      TrieHH  & $0.65$ & 1 & $0.76$ & 1 \\ \hline
      SFP ($20$) & $0.17$ & 0.853 & $0.19$ & 0.867  \\ \hline
      SFP ($80$) & $0.25$ & 0.494 & $0.325$ & 0.456 \\ \hline
    \end{tabular}
\end{minipage}
\end{figure}

% \begin{figure}[htb]
%     %\vspace{.1in}
%     \centerline{\fbox{  \includegraphics[width=0.43\textwidth]{f1_single_oov.jpg} }}
%     %\vspace{.1in}
%   \caption{F1 Score of the top K words in the single word setting of OOV case ($\delta= 1/n^2$).}
  
%   \label{fig-f1-single-oov}
% \end{figure}

% \begin{table}[htb]
% \centering
% \renewcommand{\arraystretch}{1.3}
%   \renewcommand{\tabcolsep}{1.6mm}
%     \begin{tabular}{ | c | c |  c | c | c | }
%       \hline
%       \multirow{2}{*}{}  
%       & \multicolumn{2}{|c|}{$\varepsilon = 1$} & \multicolumn{2}{|c|}{$\varepsilon = 4$} \\ \cline{2-5}
%       & Recall & Prec & Recall & Prec \\ \hline
%       TrieHH  & $0.65$ & 1 & $0.76$ & 1 \\ \hline
%       SFP ($20$) & $0.17$ & 0.853 & $0.19$ & 0.867  \\ \hline
%       SFP ($80$) & $0.25$ & 0.494 & $0.325$ & 0.456 \\ \hline
%     \end{tabular}
% \caption{Comparison of recall at $K=50$ and precision between TrieHH and SFP in the OOV setting for $\delta = \frac{1}{n^2}$ and  $T = 20, 80$ under SFP.}
% \vspace{-0.3cm}
% \label{table:precision_recall}
% \end{table}

\vspace{-0.30cm}
\paragraph{Multiple words per user: heavy hitters case}
%As in Section \ref{sec:multiple_sequences}, we use $f_i(w_j)$ to represent the local frequency of $w_j$ on $u_i$'s device. % Notice that we could have weighted the local frequencies with the number of times a $w_j$ on each user device, but we chose to cap contribution of each userto eliminate 
%We run Algorithm \ref{trie-alg-multi} on Sentiment140, choosing the same batch size as in the single words setting, $m = \gamma \sqrt{n}$, where $\gamma = \frac{e ^ {\frac{\varepsilon}{L}} - 1}{10 e ^ {\frac{\varepsilon}{L}}} \sqrt{n}$.
% By Corollary \ref{coro-fix}, Algorithm \ref{trie-alg-multi} acheives $(\varepsilon,\frac{1}{300n})$-differential privacy.
We use Sentiment140 as is for this experiment and calculate the population frequency of $w_j$ by $F(w_j) = \frac{1}{n}\sum_{i} f_i(w_j)$. Similar to the single word setting, Figure \ref{fig-freq-multi} shows the relationship between the word frequency and the discovery rate using Algorithm \ref{trie-alg-multi}. 
%We use the Monte Carlo method to run Algorithm \ref{trie-alg-multi} for 2000 times. 
Note that in the multiple words setting, it is difficult to get a non-trivial lower bound on the discovery rate of Algorithm \ref{trie-alg-multi} because such bound heavily depends on the distribution of words. 
%If a large fraction of the words share the same prefixes, the discovery rate would be high. While if the words are uniformly distributed and not any sharing prefixes, the discovery rate can be made arbitrarily low. Not surprisingly, 
Figure \ref{fig-freq-multi} shows the discovery rate and Figure \ref{fig-recall-multi} shows the recall of the top $K$ words. Observe that the top 200 words are recalled at a rate close to 1 with $\varepsilon = 4$ and $\delta < 5 \times 10^{-9}$
%We run Algorithm \ref{trie-alg-multi} for 10 times and report the mean recall and 95\% confidence intervals. Again, the recalls are higher for higher privacy budget (higher $\varepsilon$), and vice versa. 
%By comparing Figure \ref{fig-recall-single} to Figure \ref{fig-recall-multi}, observe that the recall is much better in the multiple words setting. The top 350 words can now be recalled at a rate close to 1 with a single digit $\varepsilon$.

\begin{figure}[!tbp]
  \centering
  \begin{minipage}[t]{0.48\textwidth}
    \includegraphics[width=\textwidth]{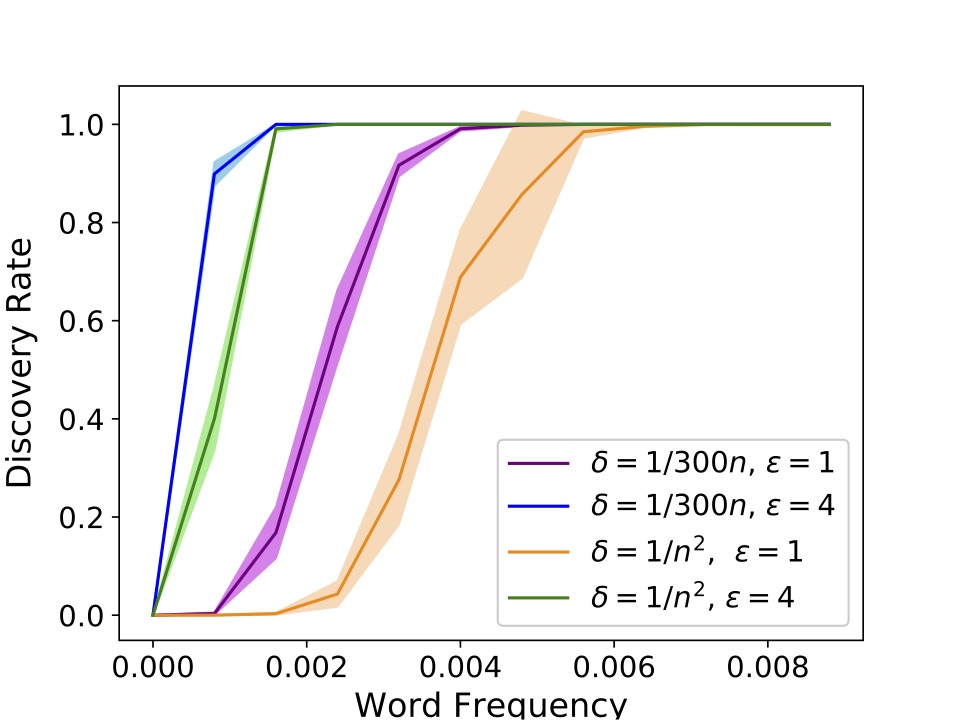}
  \caption{Sequence frequency vs. the discovery rate in the multiple words setting.}
  \label{fig-freq-multi}
  \end{minipage}
  \hfill
  \begin{minipage}[t]{0.48\textwidth}
    \includegraphics[width=\textwidth]{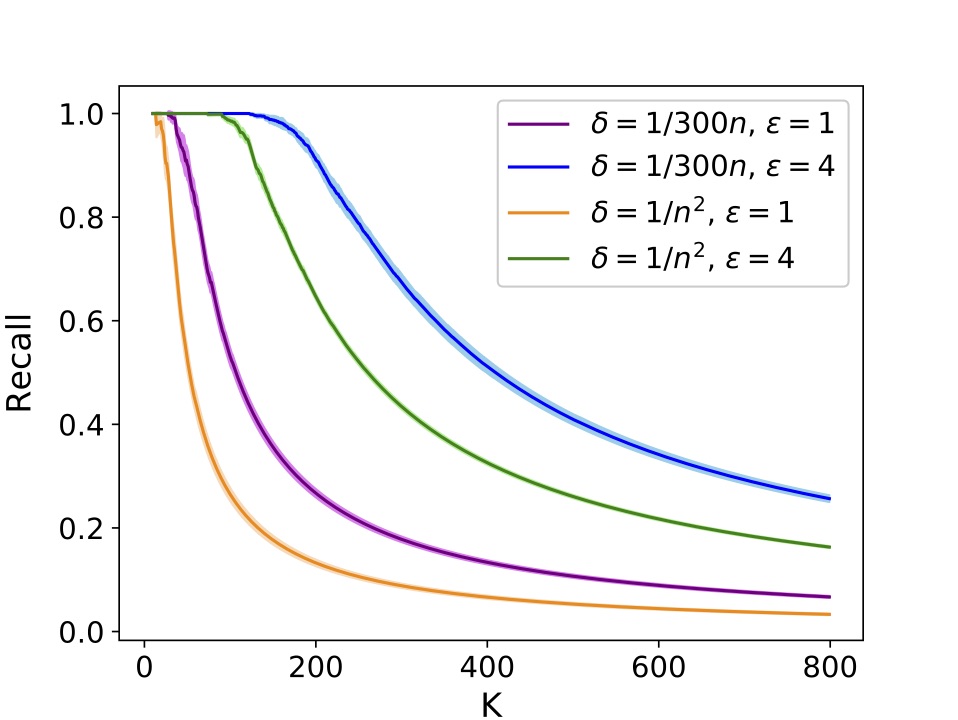}
  \caption{Recall of the top K words for different fixed $\varepsilon$ in the multiple words setting.}
  \label{fig-recall-multi}
  \end{minipage}
\end{figure}

% \begin{figure}[htb]
%   % \vspace{.1in}
%     \centerline{\fbox{  \includegraphics[width=0.43\textwidth]{freq_plot_multi.jpg} }}
%   % \vspace{.1in}
%   \caption{Sequence frequency vs. the discovery rate in the multiple words setting.}
%   \label{fig-freq-multi}
% \end{figure}

% \begin{figure}[htb]
%     %\vspace{.1in}
%     \centerline{\fbox{  \includegraphics[width=0.43\textwidth]{recall_multi.jpg} }}
%     %\vspace{.1in}
%   \caption{Recall of the top K words for different fixed $\varepsilon$ in the multiple words setting.}
%   \vspace{-0.3cm}
%   \label{fig-recall-multi}
% \end{figure}

\section{Conclusion and Open Questions}
\label{sec:conclusion}

We have introduced a novel federated algorithm for learning the frequent sequences, proved that it is inherently differentially private, investigated the trade-off between privacy and utility, and showed that it can provide excellent utility while achieving strong privacy guarantees. A significant advantage of this approach is that it eliminates the need to centralize raw data while also avoiding the harsh utility penalty of differential privacy in the local model. Many questions remain to be addressed, including (a) examining whether or not interactivity is necessary, (b) exploring secure multi-party computation and cryptographic primitives such as shuffling, threshold oblivious pseudorandom functions, and fully homomorphic encryption to provide stronger privacy guarantees, and (c) investigating the role of local plausible deniability (by allowing users to vote on wrong prefixes with small probability) and analyzing the privacy amplification gains obtained in the central model.

\bibliography{bio}
\bibliographystyle{plainnat}

%%%%%%%%%%%%%%%%%%%%%%%%%%%%%%%%%%%%%%%%%%%%%%%%%%%%%%%%%%%%%%%%%%%%%%%%%%%%%%%
%%%%%%%%%%%%%%%%%%%%%%%%%%%%%%%%%%%%%%%%%%%%%%%%%%%%%%%%%%%%%%%%%%%%%%%%%%%%%%%
% DELETE THIS PART. DO NOT PLACE CONTENT AFTER THE REFERENCES!
%%%%%%%%%%%%%%%%%%%%%%%%%%%%%%%%%%%%%%%%%%%%%%%%%%%%%%%%%%%%%%%%%%%%%%%%%%%%%%%
%%%%%%%%%%%%%%%%%%%%%%%%%%%%%%%%%%%%%%%%%%%%%%%%%%%%%%%%%%%%%%%%%%%%%%%%%%%%%%%
%\onecolumn

\newpage

\clearpage

\appendix

\begin{center}
{\Large Supplementary Material}
\end{center}

\section{Proof of Theorem \ref{thm-dp} and  Theorem \ref{thm-dp-multi-general}}

We will show that when $n \ge 10^4$, choosing $\theta \ge 10$, $\gamma \ge 1$, and $\gamma \le \frac{\sqrt{n}}{\theta + 1}$, ensures that Algorithm \ref{trie-alg} is $\left(L \ln \left(1 + \frac{1}{\frac{\sqrt{n}}{\gamma \theta} - 1}\right), \text{ }\frac{\theta-2}{(\theta-3) \theta!}\right)$-differentially private. This theorem is proved by combining two lemmas that deal with different cases of the population. In Lemma \ref{lemma-T}, we first show a bound on the ratio between $P(\mathcal{M}_i(\mathcal{D}, \theta, \gamma, T_{i-1}) = T_i)$ and $P(\mathcal{M}_i(\mathcal{D}', \theta, \gamma, T_{i-1}) = T_i)$ for any trie $T \in Range(\mathcal{M})$ that $p_i(w) \in T_i$ . This bound depends on $k$, the number of sequences that have prefix $p_i(w)$ in $\mathcal{D}'$. It is obvious that when $k = \theta - 1$, $P(\mathcal{M}_i(\mathcal{D}', \theta, \gamma, T_{i-1}) = T_i)$ must be 0, but the number of sequences having prefix $p_i(w)$ in $\mathcal{D}$ is $\theta$, so $P(\mathcal{M}_i(\mathcal{D}, \theta, \gamma, T_{i-1}) = T_i)$ is greater than 0. In this case, the ratio between them approaches infinity. On the one hand, if the number of sequences with prefix $p_i(w)$ in $\mathcal{D}'$ is already large, then an extra $p_i(w)$ in $\mathcal{D}$ only affects the probability slightly, so the ratio between $P(\mathcal{M}_i(\mathcal{D}, \theta, \gamma, T_{i-1}) = T_i)$ and $P(\mathcal{M}_i(\mathcal{D}', \theta, \gamma, T_{i-1}) = T_i)$ is small, and it could be bounded by a small $\varepsilon$. On the other hand, if the number of sequences with $p_i(w)$ in $\mathcal{D}$ is actually small, then the probability $P(\mathcal{M}_i(\mathcal{D}, \theta, \gamma, T_{i-1}) = T_i)$ is small, and could be bounded by a reasonably small $\delta$. This case is handled by Lemma \ref{lemma-k}.

We start by calculating the probability that a prefix $p$ appears at least $\theta$ times if we randomly choose $m$ users from a pool of users of size $n$, assuming that $p$ appears $W$ times in the population.

\begin{proposition}
\label{prop}
Suppose prefix $p$ appears $W$ times in a pool of $n$ users. If we select $m$ users uniformly at random from them, then the probability that prefix $p$ is appears at least $\theta$ times is
\begin{equation*}
    \frac{1}{\binom{n}{m}} \sum_{i = \theta}^{\min\{W, m\}} \binom{W}{i} \binom{n-W}{m-i}
\end{equation*}
\end{proposition}

\begin{proof}
The probability that a prefix $p$ appears $i$ times follows the hypergeometric distribution $P(i) =  \frac{1}{\binom{n}{m}} \binom{W}{i} \binom{n-W}{m-i}$. To calculate the probability that $p$ appears at least $\theta$ times in the chosen subset, we sum up the case that $p$ appears $\theta, \theta + 1, \dots, \min\{W, m\}$ times.
\end{proof}

The above probability expression will be useful in the proof of Lemma \ref{lemma-k} below, and when we investigate the privacy-utility trade-off in Section \ref{sec:single_sequence}. Also, Proposition \ref{prop:discovery_rate} is derived from Proposition \ref{prop}.

\begin{lemma}
\label{lemma-T}
$\forall T \in Range(\mathcal{M})$ such that $p_i(w) \in T_i$, $\forall i \in \{1, \dots, l\}$, assume there are $k$ users in $\mathcal{D}'$ that have prefix $p_i(w)$, and $k \ge \theta$. Then
 $P(\mathcal{M}_i(\mathcal{D}, \theta, \gamma, T_{i-1}) = T_i) \le \ ( 1 + \frac{\theta}{k - \theta + 1} )P(\mathcal{M}_i(\mathcal{D}', \theta, \gamma, T_{i-1}) = T_i)$.
\end{lemma}

% In this lemma, we show a bound on the ratio between $P(\mathcal{M}_i(\mathcal{D}, \theta, \gamma, T_{i-1}) = T_i)$ and $P(\mathcal{M}_i(\mathcal{D}', \theta, \gamma, T_{i-1}) = T_i)$. Let $C(\mathcal{X}, \mu, \theta, \gamma, T_{in}, T_{out})$ be a function that counts the number of ways to choose $\mu$ users (denote the set of chosen users as $\mathcal{X}'$) from a set of users $\mathcal{X}$. We calculate $P(\mathcal{M}_i(\mathcal{D}, \theta, \gamma, T_{i-1}) = T_i)$ by the ratio between $C(\mathcal{D}, m, \theta, \gamma, T_{i-1}, T_i)$ (the number of ways to choose $m$ users from $\mathcal{D}$, that $\mathcal{V}(\mathcal{D},  T_{in}, \theta, i)$ returns $T_{out}$) and $\binom{n}{m}$ (the number of ways to choose $m$ users from $\mathcal{D}$). Then we express $P(\mathcal{M}_i(\mathcal{D}', \theta, \gamma, T_{i-1}) = T_i)$ in the same method.

% Also, we could decompose $C(\mathcal{D}', m, \theta, \gamma, T_{i-1}, T_i)$ into two parts: not choosing $w'$ (taking all $m$ users from $\mathcal{D}' - \{w'\}$), or choosing $w'$ (taking the rest $m-1$ users from $\mathcal{D}' - \{w'\}$). Note that $\mathcal{D}$ and $\mathcal{D}'$ only differs in one user, so when we expand $P(\mathcal{M}_i(\mathcal{D}, \theta, \gamma, T_{i-1}) = T_i)$ and $P(\mathcal{M}_i(\mathcal{D}', \theta, \gamma, T_{i-1}) = T_i)$, these two formulas share some common expressions. This helps us bound the ratio.

\begin{proof}

Let $C(\mathcal{D}, \mu, \theta, \gamma, T_{in}, T_{out})$ be a function to count the number of ways to choose $\mu$ users (denote the set of chosen users as $\tilde{\mathcal{X}}$) from a set of users $\mathcal{D}$, that using Algorithm \ref{vote-alg}, $\mathcal{V}(\tilde{\mathcal{X}}, T_{in}, \theta, i) = T_{out}$. Also, we denote $C(\mathcal{D}, \mu, \theta, \gamma, T_{in}, T_{out} | p_i(w))$ as the number of ways to choose users under the same condition, given prefix $p_i(w)$ is added to $T_{out}$ in this step.

Remember $\mathcal{D}$ and $\mathcal{D}'$ differ in only one sequence $w$ and $w'$, that $|w| = l$, $|w'| = 0$. We denote $w$'s prefix of length $i$ as $p_i(w)$. For any output trie $T \in Range(\mathcal{M})$, consider the step to grow $T_{i}$ from $T_{i-1}$ by $M_i$. Let $\mathcal{Z} = \mathcal{D} - \{w\} = \mathcal{D}' - \{w'\}$. We assume there are $k$ users in $\mathcal{D}'$ that have prefix $p_i(w)$ of $w$. We denote this subset of users in $\mathcal{D}'$ as $\mathcal{W}$. Thus the set of users in $\mathcal{D}$ that have prefix $p_i(w)$ is $\mathcal{W} + \{w\}$ with size $k + 1$.

We abuse the notation to use $C(\mathcal{D}, \mu)$ instead of $C(\mathcal{D}, \mu, \theta, \gamma, T_{in}, T_{out})$, and  $C(\mathcal{D}, \mu | p_i(w))$ instead of $C(\mathcal{D}, \mu, \theta, \gamma, T_{in}, T_{out} | p_i(w))$ for fixed $\theta$, $\gamma$, $T_{in}$, $T_{out}$.

We calculate $P(\mathcal{M}_i(\mathcal{D}, \theta, \gamma, T_{i-1}) = T_i)$ by the ratio between $C(\mathcal{D}, m, \theta, \gamma, T_{i-1}, T_i)$ (how many ways to choose $m$ users from $\mathcal{D}$, that $\mathcal{V}(\mathcal{D}, T_{in}, \theta, i)$ returns $T_{out}$ and $\binom{n}{m}$ (how many ways to choose $m$ users from $\mathcal{D}$).

Also, we could separate $C(\mathcal{D}', m, \theta, \gamma, T_{i-1}, T_i)$ into two parts: not choosing $w'$ (taking all $m$ users from $\mathcal{D}' - \{w'\}$), or choosing $w'$ (taking the rest $m-1$ users from $\mathcal{D}' - \{w'\}$). Thus,

\begin{align*}
&P(\mathcal{M}_i(\mathcal{D}', \theta, \gamma, T_{i-1}) = T_i) = \frac{C(\mathcal{D}', m, \theta, \gamma, T_{i-1}, T_i)}{\binom{n}{m}} \\
&= \frac{1}{\binom{n}{m}}(C(\mathcal{D}' - \{w'\}, m) + C(\mathcal{D}' - \{w'\}, m-1)) \\
&= \frac{1}{\binom{n}{m}}(C(\mathcal{Z}, m) + C(\mathcal{Z}, m-1))
\stepcounter{equation}\tag{\theequation}\label{eq-for-multi1}
\end{align*} 

Consider $C(\mathcal{Z}, m-1, \theta, \gamma, T_{i-1}, T_i)$, because $p_i(w) \in T_i$, so there must be at least $\theta$ users in the chosen set voting for $p_i(w)$. We consider the following cases separately: choosing $\theta$ users from $\mathcal{W}$, $m-1-\theta$ users from $\mathcal{Z}-\mathcal{W}$ (note that $p_i(w) \in T_i$ is already guaranteed by choosing $\theta$ users from $\mathcal{W}$, so we consider $p_i(w)$ as a given condition here), and choosing $\theta + 1$ users from $\mathcal{W}$, $m-\theta-1$ users from $\mathcal{Z}-\mathcal{W}$, $\dots$, i.e.,
\begin{equation*}
C(\mathcal{Z}, m-1) = \sum_{i=\theta}^{\min\{k, m\}} \binom{k}{i} C(\mathcal{Z}-\mathcal{W}, m-i-1 | p_i(w))
\end{equation*}

Similarly for $C(\mathcal{Z}, m)$, not choosing $w$ (taking all $m$ users from $\mathcal{D} - \{w\}$) or choosing $w$ (taking the rest $m-1$ users from $\mathcal{D}' - \{w\}$).

\begin{align*}
C(\mathcal{Z}, m) &= \binom{k}{\theta} C(\mathcal{Z}-\mathcal{W}, m-\theta | p_i(w))\\
&+ \sum_{i = \theta+1}^ {\min\{k, m\}} \binom{k}{i} C(\mathcal{Z}-\mathcal{W}, m-i | p_i(w))
\end{align*}

Thus,
\begin{equation}
\label{eq-k}
C(\mathcal{Z}-\mathcal{W}, m-\theta | p_i(w)) \le \frac{1}{\binom{k}{\theta}} C(\mathcal{Z}, m)
\end{equation}

$C(\mathcal{D}, m, \theta, \gamma, T_{i-1}, T_i)$  could also be considered as not choosing $w$ (taking all $m$ users from $\mathcal{D} - \{w\}$) or choosing $w$ (taking the rest $m-1$ users from $\mathcal{D} - \{w\}$). But different from $\mathcal{D}'$, if $w \in \mathcal{D}$ is chosen, we can choose $\theta - 1$ to $k$ users contain prefix $p_i(w)$ from $\mathcal{Z} - \mathcal{W}$. Thus,

\begin{align*}
&C(\mathcal{D}, m, \theta, \gamma, T_{i-1}, T_i)\\
&= C(\mathcal{Z}, m)
+ \binom{k}{\theta-1} C(\mathcal{Z}-\mathcal{W}, m-\theta | p_i(w))\\
&+ \sum_{i = \theta} ^ {\min\{k, m\}} \binom{k}{i} C(\mathcal{Z}-\mathcal{W}, m-i-1 | p_i(w))\\
&= C(\mathcal{Z}, m)
+ \binom{k}{\theta-1} C(\mathcal{Z}-\mathcal{W}, m-\theta | p_i(w))\\
&+ C(\mathcal{Z}, m-1) \stepcounter{equation}\tag{\theequation}\label{eq-c}
\end{align*}

By Equation \ref{eq-c} and Inequality \ref{eq-k},
\begin{align*}
&P(\mathcal{M}_i(\mathcal{D}, \theta, \gamma, T_{i-1}) = T_i) = \frac{C(\mathcal{D}, m, \theta, \gamma, T_{i-1}, T_i)}{\binom{n}{m}} \\
&= \frac{1}{\binom{n}{m}}(C(\mathcal{Z}, m) + \binom{k}{\theta-1} C(\mathcal{Z}-\mathcal{W}, m-\theta | p_i(w))\\
&+ C(\mathcal{Z}, m-1))\\
&\le \frac{1}{\binom{n}{m}}(C(\mathcal{Z}, m) + C(\mathcal{Z}, m-1) + \frac{\binom{k}{\theta-1}}{\binom{k}{\theta}}C(\mathcal{Z}, m))\\
&= \frac{1}{\binom{n}{m}}(C(\mathcal{Z}, m) + C(\mathcal{Z}, m-1) \\
&+ \frac{\theta}{k - \theta + 1}C(\mathcal{Z}, m)) \\
&\le (1 + \frac{\theta}{k - \theta + 1}) \frac{1}{\binom{n}{m}}(C(\mathcal{Z}, m) + C(\mathcal{Z}, m-1))\\
&= (1 + \frac{\theta}{k - \theta + 1})P(\mathcal{M}_i(\mathcal{D}', \theta, \gamma, T_{i-1}) = T_i)
\end{align*}

\end{proof}

Suppose there are $k$ users has prefix $p_1(w)$. In Lemma \ref{lemma-k}, we show that when $k \le \frac{\sqrt{n}}{\gamma} - 1$, $P(p_1(w) \in \mathcal{M}_1(\mathcal{D}, \theta, \gamma)) \le \frac{\theta - 2}{(\theta - 3) \theta !}$. This means when $k$ is small, the probability that $p_1(w) \in \mathcal{M}_1(\mathcal{D}, \theta, \gamma)$ is small, so it could be bounded by a small $\delta$. And it is the same for the $i^th$ round that when there are $k$ users has prefix $p_i(w)$. If $k \le \frac{\sqrt{n}}{\gamma} - 1$, then $P(p_{i-1}(w) \in \mathcal{M}(\mathcal{D}, \theta, \gamma)|p_{i-2}(w) \in \mathcal{M}(\mathcal{D}, \theta, \gamma)) \le \frac{\theta-2}{(\theta - 3) \theta !}$.

\begin{lemma}
\label{lemma-k}
Consider the step to grow $T_{i}$ from $T_{i-1}$ by $\mathcal{M}_i$. We assume there are $k$ users in $\mathcal{D}'$ that have prefix $p_i(w)$ of $w$. Then there are $k+1$ users in $\mathcal{D}$ that have prefix $p_i(w)$. When $k \le \frac{\sqrt{n}}{\gamma} - 1$, $4 \le \theta \le \sqrt{n}$, $\gamma \ge 1$, $P(p_{i}(w) \in \mathcal{M}(\mathcal{D}, \theta, \gamma)|p_{i-1}(w) \in \mathcal{M}(\mathcal{D}, \theta, \gamma)) \le \frac{\theta-2}{(\theta - 3) \theta !}$. For the first step, $P(p_1(w) \in \mathcal{M}_1(\mathcal{D}, \theta, \gamma)) \le \frac{\theta - 2}{(\theta-3) \theta !}$.
\end{lemma}

% \begin{proofsketch}
% This lemma shows that when the number of a given prefix $k \le \frac{\sqrt{n}}{\gamma} - 1$, then the probability that this prefix will be added to the trie could be bounded by $\frac{1}{100n}$. Intuitively, when $k$ is small, the chance that the prefix grows on the trie is small. We can calculate the probability by Proposition \ref{prop}. We first prove that the probability is upper bounded by the sum of a geometric sequence, and then bound the first item of the sequence to finish the proof.
% \end{proofsketch}

\begin{proof}

First $P(p_1(w) \in \mathcal{M}(\mathcal{D}, \theta, \gamma)) \le P(p_1(w) \in \mathcal{M}_1(\mathcal{D}, \theta, \gamma, T_0))$. To calculate $P(p_1(w) \in \mathcal{M}_1(\mathcal{D}, \theta, \gamma, T_0))$, we consider the cases of choosing $\theta$ to $k+1$ users voting for $p_i(w)$ separately, By Proposition \ref{prop},

\begin{align*}
&P(p_1(w) \in \mathcal{M}_1(\mathcal{D}, \theta, \gamma, T_0))\\
&= \frac{1}{\binom{n}{m}} \sum_{i = \theta}^{\min\{k+1, m\}} \binom{k+1}{i} \binom{n-k-1}{m -i}
\end{align*}

Note that when $k+1 < \theta$, $P(p_1(w) \in \mathcal{M}_1(\mathcal{D}, \theta, \gamma, T_0)) = 0$, so we only consider the case that $k+1 \ge \theta$. The sum of the array above could be upper bounded by the sum of a geometric sequence. We know that $k \le \frac{\sqrt{n}}{\gamma} - 1$, $m=\gamma \sqrt{n}$. Consider the ratio between the first two items,\\
\begin{align}
  \frac{ \binom{k+1}{\theta+1} \binom{n-k-1}{m-\theta-1} }{ \binom{k+1}{\theta} \binom{n-k-1}{m-\theta} } &= \frac{ (k-\theta+1) (m-\theta) }{ (\theta+1) (n-k-m+\theta) } \nonumber\\
   		  &\le \frac{ (\frac{\sqrt{n}}{\gamma}-\theta) (\gamma \sqrt{n}-\theta) }{ (\theta+1) (n-\frac{\sqrt{n}}{\gamma}+1-\gamma \sqrt{n}+\theta) } \nonumber\\
   		  &\le \frac{ n }{ (\theta + 1) (n - (\gamma + \frac{1}{\gamma}) \sqrt{n} +1+\theta) } \nonumber\\
   		  &\le \frac{1}{ (\theta + 1) (1 - \frac{\gamma + \frac{1}{\gamma}}{\sqrt{n}}) } \label{eq-1}
  \end{align}

We denote $\frac{1}{ (\theta + 1) (1 - \frac{\gamma + \frac{1}{\gamma}}{\sqrt{n}}) }$ as $r_1$. Because $k \le \frac{\sqrt{n}}{\gamma} - 1$ and $k+1 \ge \theta$, so $\gamma \le \frac{\sqrt{n}}{k+1} \le \frac{\sqrt{n}}{\theta}$. We know that $\gamma \ge 1$, then $\gamma + \frac{1}{\gamma} \le \gamma + 1 \le \frac{\sqrt{n}}{\theta}+1$. And also $\theta \le \sqrt{n}$, so $\gamma + \frac{1}{\gamma} \le \frac{\sqrt{n}}{\theta}+1 \le 2\frac{\sqrt{n}}{\theta}$. Now we are able to bound $r_1$:

  \begin{align*}
    r_1 &= \frac{1}{ (\theta + 1) (1 - \frac{\gamma + \frac{1}{\gamma}}{\sqrt{n}}) }\\
        &\le \frac{1}{ (\theta + 1) (1 - \frac{\frac{2\sqrt{n}}{\theta}}{\sqrt{n}} ) }\\
        &\le \frac{1}{ (\theta + 1) (1 - \frac{2}{\theta} ) }\\
        &\le \frac{1}{ \theta (1 - \frac{2}{\theta} ) }\\
        &= \frac{1}{ \theta - 2}
  \end{align*}
  
Note that when $\theta \ge 4$, $r_1 < 1$. 

Now we'll show that the ratio between adjacent items is decreasing. Consider the ratio between any two adjacent items $\binom{k+1}{\theta+i+1} \binom{n-k-1}{m-\theta-i-1}$ and $\binom{k+1}{\theta+i} \binom{n-k-1}{m-\theta-i}$,

\begin{align*}
\frac{\binom{k+1}{\theta+i+1} \binom{n-k-1}{m-\theta-i-1}}{\binom{k+1}{\theta+i} \binom{n-k-1}{m-\theta-i}} &= \frac{(k-\theta-i+1)(m-\theta-i)}{(\theta+i+1)(n-k-m+\theta+i)} \\
		   &\le \frac{ (k-\theta+1) (m-\theta) }{ (\theta+1) (n-k-m+\theta) } \\
		   &= \frac{ \binom{k+1}{\theta+1} \binom{n-k-1}{m-\theta-1} }{ \binom{k+1}{\theta} \binom{n-k-1}{m-\theta} } \\
		   &\le r_1
\end{align*}

Thus,

\begin{align*}
 &P(p_1(w) \in \mathcal{M}_1(\mathcal{D}, \theta, \gamma, T_0))\\
 &= \frac{1}{\binom{n}{m}} \sum_{i = \theta}^{k+1}  \binom{k+1}{i} \binom{n-k-1}{m - i}\\
 &\le \frac{1}{\binom{n}{m}} (\sum_{i = 0}^{k+1-\theta} r_1 ^ i) \binom{k+1}{\theta} \binom{n-k-1}{m-\theta} \\
 &\le \frac{1}{\binom{n}{m}} \frac{1}{1 - r_1} \binom{k+1}{\theta} \binom{n-k-1}{m-\theta}\\
 \end{align*}
 
The last line follows because $r_1 < 1$.

 When $k \le \frac{\sqrt{n}}{\gamma} - 1$,
   \begin{align}
   &P(p_1(w) \in \mathcal{M}_1(\mathcal{D}, \theta, \gamma, T_0)) \nonumber\\
   &\le \frac{1}{\binom{n}{m}} \frac{1}{1 - r_1} \binom{k+1}{\theta} \binom{n-k-1}{m-\theta} \nonumber\\
   &\le \frac{1}{(1 - r_1)\theta !} \times \prod_{i=0}^{\theta-1}(k+1-i) \times \prod_{i=0}^{\theta-1}(m-i)\times \frac{1}{\prod_{i=0}^{\theta-1}(n-i)} \nonumber\\
   &\le \frac{1}{(1 - r_1)\theta !} \times \prod_{i=0}^{\theta-1}(\frac{\sqrt{n}}{\gamma}-i) \times \prod_{i=0}^{\theta-1}(\gamma \sqrt{n}-i)\times \frac{1}{\prod_{i=0}^{\theta-1}(n-i)} \nonumber\\
   &\le \frac{1}{(1 - r_1)\theta !} \times \prod_{i=0}^{\theta-1} \frac{n - (\frac{1}{\gamma} + \gamma) \sqrt{n} i + i ^ 2}{n-i} \nonumber\\
   &\le \frac{1}{(1 - r_1)\theta !} \nonumber
   \end{align}

  Thus,
  \begin{equation*}
    P(p_1(w) \in \mathcal{M}_1(\mathcal{D}, \theta, \gamma, T_0)) \le \frac{1}{(1 - r_1)\theta !} \le \frac{\theta - 2}{(\theta  - 3)\theta !}
  \end{equation*}

 We could get the same upper bound for $P(p_{i-1}(w) \in \mathcal{M}(\mathcal{D}, \theta, \gamma)|p_{i-2}(w) \in \mathcal{M}(\mathcal{D}, \theta, \gamma))$ when there are $k+1$ users containing prefix $p_i(w)$, because it is also a one step voting process to determine if $p_i(w)$ will grow on the trie.

\end{proof}

\subsection{Proof of Theorem \ref{thm-dp}}
\begin{proof}
By definition, algorithm $\mathcal{M}$ has $(\varepsilon, \delta)$-differential privacy means that, $\forall \mathcal{S} \subseteq \text{Range}(\mathcal{M})$,

\begin{equation}
\label{thm-dp-eq1}
P(\mathcal{M}(\mathcal{D}) \in \mathcal{S}) \le e^\varepsilon P(\mathcal{M}(\mathcal{D}') \in \mathcal{S}) + \delta
\end{equation}

When we choose the same fixed $\theta$ and $\gamma$ for a certain $n$, we abuse the notation to use $\mathcal{M}(\mathcal{D})$ and $\mathcal{M}_i(\mathcal{D}, T)$ instead of $\mathcal{M}(\mathcal{D}, \theta, \gamma)$ and $\mathcal{M}_i(\mathcal{D}, \theta, \gamma, T)$.

Suppose $|w| = l$. We decompose $\mathcal{S}$ into $\mathcal{S} = \mathcal{S}_0 \cup \mathcal{S}_1 \cup \dots \mathcal{S}_l$. $\mathcal{S}_0$ is the subset of $\mathcal{S}$ that contains no prefix of $w$, $\mathcal{S}_1$ is the subset of $\mathcal{S}$ that contains only $p_1(w)$,  $\mathcal{S}_i$ is the subset of $\mathcal{S}$ that only $p_1(w)$ to $p_i(w)$ of $w$. Formally, $\mathcal{S}_0 = \{ T \in \mathcal{S} | p_i(w) \notin T, \text{ for } i = 1, 2, \dots, l \}$ and $\mathcal{S}_i = \{ T \in \mathcal{S} | p_1(w), \dots, p_i(w) \in T \text{ and } p_{i+1}, \dots, p_l \notin T \}$ for $i = 1, 2, \dots, l$. Then Inequality \ref{thm-dp-eq1} is equivalent to,

\begin{equation*}
\sum\limits_{i=0}^l P(\mathcal{M}(\mathcal{D}) \in \mathcal{S}_i) \le e^\varepsilon \sum\limits_{i=0}^l P(\mathcal{M}(\mathcal{D}') \in \mathcal{S}_i) + \delta
\end{equation*}

Because the tries in $\mathcal{S}_0$ do not have any node in the path of $w = (root, c1, c2, \dots, c_l)$,
\begin{equation*}
P(\mathcal{M}(\mathcal{D}) \in \mathcal{S}_0) \le P(\mathcal{M}(\mathcal{D}') \in \mathcal{S}_0)
\end{equation*}

 We define $\mathcal{R}_i = \{ T \in \mathcal{R} | p_i(w) \in T \}$. Note that different from $\mathcal{S}_i$, $\mathcal{R}_i$ contains all possible tries that contain $p_i(w)$ (including those contain $p_{i+1}, p_{i+2}, \dots$). Thus, $\mathcal{R}_l \subseteq \mathcal{R}_{l-1} \subseteq \dots \subseteq \mathcal{R}_2 \subseteq \mathcal{R}_1$, therefore $P(\mathcal{M}(\mathcal{D}) \in \mathcal{R}_1) \ge P(\mathcal{M}(\mathcal{D}) \in \mathcal{R}_2) \ge \dots \ge P(\mathcal{M}(\mathcal{D}) \in \mathcal{R}_l)$. Let $j$ be the smallest index that $P(\mathcal{M}(\mathcal{D}) \in \mathcal{R}_j) \le \frac{\theta - 2}{(\theta  - 3)\theta !}$, if such $j$ exists. Then,

\begin{equation*}
\sum_{i=j}^l P(\mathcal{M}(\mathcal{D}) \in \mathcal{S}_i) \le P(\mathcal{M}(\mathcal{D}) \in \mathcal{R}_j) \le \frac{\theta - 2}{(\theta  - 3)\theta !}
\end{equation*}

For indexes $i < j$, we know that $P(\mathcal{M}(\mathcal{D}) \in \mathcal{R}_i) > \frac{\theta - 2}{(\theta  - 3)\theta !}$. For any $i < j$:

\begin{equation*}
P(\mathcal{M}(\mathcal{D}) \in \mathcal{R}_i) = P(p_i(w) \in \mathcal{M}(\mathcal{D})) \\
= P(p_1(w) \in \mathcal{M}(\mathcal{D}))
\times \prod_{j=1} ^ {i-1} P(p_{j+1} \in \mathcal{M}(\mathcal{D})|p_j \in \mathcal{M}(\mathcal{D}))
\end{equation*}

Because $P(\mathcal{M}(\mathcal{D}) \in \mathcal{R}_i) > \frac{\theta - 2}{(\theta  - 3)\theta !}$, it must be the case that every term on the right hand side of the equation above is greater than $ \frac{\theta - 2}{(\theta  - 3)\theta !}$, i.e., $P(p_1(w) \in \mathcal{M}(\mathcal{D})) \ge \frac{\theta-2}{(\theta-3)\theta!}$, and $P(p_{i-1}(w) \in \mathcal{M}(\mathcal{D})|p_{i-2}(w) \in \mathcal{M}(\mathcal{D})) \ge \frac{\theta-2}{(\theta-3)\theta!}$, $\forall i \in \{2, \dots, j-1\}$. By  Lemma \ref{lemma-k}, $k_i > \frac{\sqrt{n}}{\gamma} - 2$, because $k_i$ is integer, $k_i \ge \frac{\sqrt{n}}{\gamma} - 1 \ge \theta$ (because $\gamma \le \frac{\sqrt{n}}{\theta + 1}$). $\forall i \in \{1, \dots, j-1\}$.

For each $T \in \mathcal{S}_i$, $i \in \{1, \dots, j-1\}$,

\begin{equation*}
P(\mathcal{M}(\mathcal{D}') = T) = \prod_{b=1}^i P(\mathcal{M}_b(\mathcal{D}', T_{b-1}) = T_b)
\end{equation*}

Because $k_i > \frac{\sqrt{n}}{\gamma} - 2 \ge \theta$ for all $i \in \{1, \dots, j-1\}$, by Lemma \ref{lemma-T},

\begin{align*}
&P(\mathcal{M}(\mathcal{D}) = T) \\
&= \prod_{b=1}^i P(\mathcal{M}_b(\mathcal{D}, T_{b-1}) = T_b) \\
&\le \prod_{b=1}^i (1 + \frac{\theta}{k_b-\theta+1}) P(\mathcal{M}_b(\mathcal{D}', T_{b-1}) = T_b)\\
&\le (1 + \frac{\theta}{\frac{\sqrt{n}}{\gamma} - 1 - \theta +1}) ^ i P(\mathcal{M}(\mathcal{D}') = T)\\
&= (1 + \frac{1}{\frac{\sqrt{n}}{\gamma \theta} - 1}) ^ i P(\mathcal{M}(\mathcal{D}') = T)
\end{align*}

Sum up for all $T \in \mathcal{S}_1 \cup \mathcal{S}_2 \cup \dots \cup \mathcal{S}_{j-1}$,

\begin{align*}
&P(\mathcal{M}(\mathcal{D}) \in \mathcal{S}_1 \cup \mathcal{S}_2 \cup \dots \cup \mathcal{S}_{j-1}) \\
&\le (1 + \frac{1}{\frac{\sqrt{n}}{\gamma \theta} - 1}) ^ {j-1} P(\mathcal{M}(\mathcal{D}') \in \mathcal{S}_1 \cup \mathcal{S}_2 \cup \dots \cup \mathcal{S}_{j-1})
\end{align*}

\begin{align*}
&P(\mathcal{M}(\mathcal{D}) \in \mathcal{S}) \\
&= \sum_{i = 1}^{j-1} P(\mathcal{M}(\mathcal{D}) \in \mathcal{S}_i) + \sum_{i = j}^l P(\mathcal{M}(\mathcal{D}) \in \mathcal{S}_i) \\
&\le (1 + \frac{1}{\frac{\sqrt{n}}{\gamma \theta} - 1}) ^ {j-1} \sum_{i = 1}^{j-1} P(\mathcal{M}(\mathcal{D}') \in \mathcal{S}_i) + \frac{\theta-2}{(\theta-3)\theta!} \\
&\le (1 + \frac{1}{\frac{\sqrt{n}}{\gamma \theta} - 1}) ^ l \sum_{i = 1}^l P(\mathcal{M}(\mathcal{D}') \in \mathcal{S}_i) + \frac{\theta-2}{(\theta-3)\theta!}  \\
&= (1 + \frac{1}{\frac{\sqrt{n}}{\gamma \theta} - 1}) ^ l P(\mathcal{M}(\mathcal{D}') \in \mathcal{S}) + \frac{\theta-2}{(\theta-3)\theta!} \\
&= (1 + \frac{1}{\frac{\sqrt{n}}{\gamma \theta} - 1}) ^ L P(\mathcal{M}(\mathcal{D}') \in \mathcal{S}) + \frac{\theta-2}{(\theta-3)\theta!}
\end{align*}

\end{proof}

\subsection{Proof of Theorem \ref{thm-dp-multi-general} }
\begin{proof}
Suppose $\mathcal{D}$ and $\mathcal{D}'$ are user-level neighboring datasets. Without loss of generality, assume $\mathcal{D} = \{D_1, D_2, \dots, D_n \}$ and $\mathcal{D}' = \{D_1', D_2, \dots, D_n \}$. Let $\tilde{M}$ denote the process to first randomly selects 1 record per user (deterministically or randomly) and then applies $M$ on the sampled dataset of size $n$. 

Because $M$ satisfies $(\varepsilon, \delta)$ record level DP, we know that for any record level neighboring datasets $d$ and $d'$, and $\forall \mathcal{S} \subseteq \text{Range}(M)$,

$$P(M(d) \in \mathcal{S}) \le e^{\varepsilon} \times P(M(d') \in \mathcal{S}) + \delta$$

For any record level neighboring datasets $d$ and $d'$, without loss of generality, we will foucs on neighboring datasets that differ in the first record: $d = d_1 d_2 \dots d_n$ and $d' = d_1' d_2 \dots d_n$.

Our goal is to prove that $\forall \mathcal{S} \subseteq \text{Range}(\tilde{M})$,  $P(\tilde{M}(D) \in \mathcal{S}) \le e^{\varepsilon} \times P(\tilde{M}(D') \in \mathcal{S}) + \delta$. Denote $P(d|D)$ as the probability of sampling $d$ from $D$, then we can write $P(\tilde{M}(D) \in \mathcal{S})$ as:

\begin{align*}
    P(\tilde{M}(D) \in \mathcal{S}) 
    &= \sum_{d} P(M(d) \in \mathcal{S}) \times P(d|D)\\
    &= \sum_{d_1 d_2 \dots d_n} P(M(d_1 d_2 \dots d_n) \in \mathcal{S}) \times P(d_1 d_2 \dots d_n|D)\\
    &= \sum_{d_1} \sum_{d_2 \dots d_n} P(M(d_1 d_2 \dots d_n) \in \mathcal{S}) \times P(d_1 | D_1) \times P(d_2 \dots d_n | D_2 \dots D_n)\\
    &= \sum_{d_2 \dots d_n} [\sum_{d_1} P(M(d_1 d_2 \dots d_n) \in \mathcal{S}) \times P(d_1 | D_1)] \times P(d_2 \dots d_n | D_2 \dots D_n)\\
\end{align*}

Now we bound $\sum_{d_1} P(M(d_1 d_2 \dots d_n) \in \mathcal{S}) \times P(d_1 | D_1)$, and then finish the proof. For any $d_1$ and $d_1'$, we know that $P(M(d_1 d_2 \dots d_n) \in \mathcal{S}) \le e^{\varepsilon} \times P(M(d_1' d_2 \dots d_n) \in \mathcal{S}) + \delta$. Thus, for a fixed $d_1$ and arbitrary $d_1'$,

\begin{align*}
    &\sum_{d_1} P(M(d_1 d_2 \dots d_n) \in \mathcal{S}) \times P(d_1 | D_1)\\
    &\le \sum_{d_1} (e^{\varepsilon} \times P(M(d_1' d_2 \dots d_n) \in \mathcal{S}) + \delta) \times P(d_1 | D_1)\\
    &= e^{\varepsilon} \times P(M(d_1' d_2 \dots d_n) \in \mathcal{S}) + \delta
\end{align*}

Multiply both sides by $P(d_1'|D_1')$, and then sum over all $d_1'$, 

$$ \sum_{d_1} P(M(d_1 d_2 \dots d_n) \in \mathcal{S}) \times P(d_1 | D_1)
    \le e^{\varepsilon} \times (\sum_{d_1'} P(M(d_1' d_2 \dots d_n) \in \mathcal{S}) \times P(d_1'|D_1')) + \delta $$
    
Now we finish the proof using the inequality above,

\begin{align*}
    P(\tilde{M}(D) \in \mathcal{S}) 
    &= \sum_{d_2 \dots d_n} [\sum_{d_1} P(M(d_1 d_2 \dots d_n) \in \mathcal{S}) \times P(d_1 | D_1)] \times P(d_2 \dots d_n | D_2 \dots D_n) \\
    &\le \sum_{d_2 \dots d_n} (e^{\varepsilon} \times (\sum_{d_1'} P(M(d_1' d_2 \dots d_n) \in \mathcal{S}) \times P(d_1'|D_1')) + \delta) \times P(d_2 \dots d_n | D_2 \dots D_n) \\
    &\le e^{\varepsilon} \sum_{d_2 \dots d_n} \sum_{d_1'} P(M(d_1' d_2 \dots d_n) \in \mathcal{S}) \times P(d_1'|D_1') \times P(d_2 \dots d_n | D_2 \dots D_n)  + \delta\\
    &= e^{\varepsilon} \times P(\tilde{M}(D') \in \mathcal{S})  + \delta
\end{align*}

\end{proof}

\section{Proof for Corollaries}

% \subsection{Proof of Corollary \ref{coro-theta}}

% \begin{proof}

% We choose $\theta = \lceil \log_{10} n + 6 \rceil$. When $n = 10^4$, $\theta = 10$, and if $n$ is greater than $10^4$, it is easy to see that $\theta!$ increase faster than $n$. Formally, when $n$ increase by 10 times, $\theta$ increase by 1, and $\theta!$ increase by more than 10 times. Thus, for $n \ge 10^4$, 

%   \begin{equation*}
%   \theta ! \ge \frac{n}{10 ^ 4} * 10 ! = \frac{n}{10 ^ 4} * 3.6 * 10 ^ 6 = 360n
%   \end{equation*}

%   Also when $\theta \ge 10$, $\frac{\theta-2}{\theta-3} \le \frac{8}{7}$, then,
  
%   \begin{equation*}
%   \frac{\theta-2}{(\theta-3) \theta !} \le \frac{1}{300n}
%   \end{equation*}
  
% \end{proof}

% \subsection{Proof for Corollary \ref{coro-fix}}

% \begin{proof}
%   By Corollary \ref{coro-theta}, solve $L\ln(1 + \frac{1}{\frac{\sqrt{n}}{\gamma \theta} - 1}) \le \varepsilon$, we get $ \gamma \le \frac{e ^ {\frac{\varepsilon}{L}} - 1}{\theta e ^ {\frac{\varepsilon}{L}}} \sqrt{n}$. Theorem \ref{thm-dp} requires $\gamma \le \frac{\sqrt{n}}{\theta + 1}$, this is satisfied by  $\varepsilon \le L \ln(\theta + 1)$. 
% \end{proof}

\subsection{Proof for Corollary \ref{coro-theta-lambert}}

\begin{proof}

First get $\theta$ by standard calculation: $\theta = \text{max} \{ 10, \lceil e^{W(C) + 1} - \frac{1}{2} \rceil \}$, where $W$ is the Lambert $W$ function \citep{corless1996lambertw} and $C = (\ln \frac{8}{7\sqrt{2\pi}\delta})/e$. Then solve $L\ln(1 + \frac{1}{\frac{\sqrt{n}}{\gamma \theta} - 1}) \le \varepsilon$, we get $ \gamma \le \frac{e ^ {\frac{\varepsilon}{L}} - 1}{\theta e ^ {\frac{\varepsilon}{L}}} \sqrt{n}$. Theorem \ref{thm-dp} requires $\gamma \le \frac{\sqrt{n}}{\theta + 1}$, this is satisfied by  $\varepsilon \le L \ln(\theta + 1)$. 

When $n \ge 10^4$, choose $\theta = \lceil \log_{10} n + 6 \rceil$. When $n = 10^4$, $\theta = 10$, and if $n$ is greater than $10^4$, it is easy to see that $\theta!$ increase faster than $n$. Formally, when $n$ increase by 10 times, $\theta$ increase by 1, and $\theta!$ increase by more than 10 times. Thus, for $n \ge 10^4$, 

  \begin{equation*}
   \theta ! \ge \frac{n}{10 ^ 4} * 10 ! = \frac{n}{10 ^ 4} * 3.6 * 10 ^ 6 = 360n
  \end{equation*}

  Also when $\theta \ge 10$, $\frac{\theta-2}{\theta-3} \le \frac{8}{7}$, then,
  
  \begin{equation*}
   \frac{\theta-2}{(\theta-3) \theta !} \le \frac{1}{300n}
  \end{equation*}
  
\end{proof}

\section{Implementation of SFP}
\label{sec:sfp}
In this section we provide a description of the full algorithm of SFP in \citep{apple2017} for completeness and give the parameters of our implementation (we use the default parameters in \citep{apple2017} when provided in the paper). 

SFP is based on Count Mean Sketch (CMS), which contains both client-side and server-side computations. On the client side, a string is mapped to a domain of size $m$ by one of $k$ three-wise independent hash functions. Then the client submit the result with random noise (depends on $\varepsilon$ to achieve $\varepsilon$ local DP) to the server. The server gathers results from all the clients and compute the heavy hitters. In our implementation, $m=1024$ and $k=2048$.

First we introduce the basis client side encoding algorithm $\mathcal{A}_{\text{client-CMS}}$. On the client side, first sample $j$ uniformly at random from $[k]$. Then construct a vector $v$ of length $m$, with $v_{h_j(d)} = -1$ and other elements in $v$ are all 1. After that, sample vector $b \in \{-1, +1\}^m$, where $b_l$ is i.i.d. and $Pr[b_l = +1] = \frac{e^{\varepsilon/2}}{e^{\varepsilon/2}+1}$. Finally, the client returns $\tilde{v} = (v_1b_1, \dots, v_mb_m)$ and index $j$. After receiving all the noisy hashed values from the clients, the server construct a sketch matrix, and for each element $d$, we can estimate the frequency of $d$ by a frequency oracle using the sketch matrix. This server-side algorithm is denoted as $\mathcal{A}_{\text{server-CMS}}$.

We consider strings of length up to 10 (as default in the paper) by padding shorter strings with $
\$$ and truncating longer strings to length 10. The full SFP algorithm also contains client side and server side algorithms. There are $k = 2048$ three-wise independent hash functions with domain size of $m = 1024$, and a hash function $h$ with domain size $256$ shared by the server and clients. Also, there is a threshold parameter $T$ (we used $T=20$ and $T=80$ in our experiments). Then apply algorithm $\mathcal{A}_{\text{client-SFP}}$ on each client, send all the results to the server and apply $\mathcal{A}_{\text{server-SFP}}$ on the server side for the final result.

In the client side algorithm $\mathcal{A}_{\text{client-SFP}}$, suppose a client holds string $s$. First sample $l$ uniformly at random from $\{1, 3, 5, 7, 9\}$, then set $r = h(s)||s[l:l+1]$. Finally, return $\mathcal{A}_{\text{server-CMS}}(r)$ and $l$ to the server. 

In the server side algorithm $\mathcal{A}_{\text{server-SFP}}$, for each $l \in \{1, 3, 5, 7, 9\}$, create sketch matrix $M_l$ by the results from the set of users submitting index $l$ and construct the frequency oracle $\tilde{f_l}$ accordingly. Also for each $l \in \{1, 3, 5, 7, 9\}$, calculate $Q_l$, which is the $T$ tuples with the largest counts $\tilde{f_l}(w||s)$ for $s \in \Omega^2$ where $\Omega$ is the 26 lowercase English letters and $w \in [256]$. For each $w \in [256]$, we form the Cartesian product of terms in $Q(w) = \{q_1 || \dots || q_9: w||q_l \in Q_l \text{ for } l \in \{1, 3, 5, 7, 9\}$. Finally return the union of all $Q(w)$ as the result heavy hitters.

\section{Additional Discussions}

\subsection{Time, Space and Communication Complexity Analysis of TrieHH}
\label{sec:comm_cost}

\paragraph{Time Complexity} Running time on the server side is $O(m)$ for each round, so the total running time is $O(mL)$. For the running time on the user side, suppose each user has at most $Z$ words, then searching for a certain prefix cost $O(ZlogZ)$. Because there are at most $\frac{m}{\theta}$ node in each level, searching for all the prefixes in this round cost at most $O(\frac{1}{\theta}mZlogZ)$. Thus the total running time for each user is $O(\frac{mZL}{\theta}logZ)$.

\paragraph{Space Complexity} Space complexity on both the server and user side is the size of the trie. Because there are at most $\frac{m}{\theta}$ node in each level, and there are at most $L$ levels except the root node, total space complexity is $O(L\frac{m}{\theta})$.

\paragraph{Communication Cost} The worst case communication cost for round $i$: $\frac{m^2}{\theta} \times C \times i$, where $C$ is the cost to communicate a node in the trie. This is because there are at most $\frac{m}{\theta}$ length $i$ paths in the trie at round i, and the server need to update the current trie with $m$ users. The algorithm runs for at most $L$ rounds, thus the total communication cost is at most $\sum_i \frac{m^2}{\theta} \times C \times i = \frac{m^2 L(L+1) C}{2\theta}$.

% Worst case communication cost for round i: (m^2/theta)*C*i 
% Worst case communication: sum_i (m^2/theta)*C*i = (m^2/theta)*C*(L(L+1)/2) 

\subsection{Heavy Hitters Lists}
\label{sec:discovered-oov}

We provide the list of the top 200 heavy hitters in the Sentiment140 dataset, and in the OOV dataset after we filter out the words in the dictionary.

\textbf{Top 100 heavy hitters with frequencies in the Sentiment140 dataset}

\begin{lstlisting}
{'the': '0.1028', 'you': '0.0360', 'and': '0.0308', 'just': '0.0209', "i'm": '0.0169', 'for': '0.0143', 'have': '0.0107', 'going': '0.0086', 'not': '0.0074', 'that': '0.0073', 'was': '0.0069', 'good': '0.0062', 'work': '0.0056', "it's": '0.0055', 'this': '0.0053', 'watching': '0.0052', 'back': '0.0051', 'got': '0.0049', 'with': '0.0048', 'had': '0.0048', 'love': '0.0047', 'really': '0.0047', "can't": '0.0046', 'has': '0.0045', 'but': '0.0043', 'miss': '0.0039', 'still': '0.0039', 'its': '0.0037', 'want': '0.0036', 'getting': '0.0035', 'day': '0.0035', "don't": '0.0033', 'happy': '0.0033', 'what': '0.0032', 'now': '0.0032', 'why': '0.0031', 'lol': '0.0031', 'home': '0.0031', 'wish': '0.0030', 'today': '0.0030', 'all': '0.0029', 'new': '0.0029', 'off': '0.0028', 'need': '0.0028', 'your': '0.0028', 'hate': '0.0026', 'sad': '0.0026', 'last': '0.0026', 'think': '0.0025', 'trying': '0.0025', 'out': '0.0025', 'get': '0.0025', 'hey': '0.0024', 'working': '0.0023', 'like': '0.0023', 'finally': '0.0022', 'too': '0.0022', 'well': '0.0022', 'about': '0.0022', 'one': '0.0021', 'will': '0.0021', 'thanks': '0.0021', 'very': '0.0021', 'are': '0.0021', 'feel': '0.0020', 'cant': '0.0020', 'time': '0.0020', 'bored': '0.0020', 'feeling': '0.0019', 'omg': '0.0019', 'having': '0.0018', 'tired': '0.0018', 'her': '0.0018', 'ugh': '0.0018', 'more': '0.0017', 'waiting': '0.0017', 'missing': '0.0016', 'sitting': '0.0016', 'twitter': '0.0016', 'haha': '0.0016', 'listening': '0.0016', 'how': '0.0016', 'wants': '0.0016', 'great': '0.0015', 'wow': '0.0015', 'sick': '0.0014', 'they': '0.0014', 'know': '0.0014', 'can': '0.0014', 'night': '0.0014', 'another': '0.0014', 'morning': '0.0014', 'damn': '0.0014', '@mileycyrus': '0.0014', 'way': '0.0014', 'yay': '0.0014', 'dont': '0.0014', 'looking': '0.0013', 'some': '0.0013', 'she': '0.0013'}
\end{lstlisting}

\textbf{Top 100 heavy hitters with frequencies in the OOV dataset generated from Sentiment140}

\begin{lstlisting}
{'dont': '0.011741', 'thats': '0.006008', 'didnt': '0.004292', 'sooo': '0.004023', 'awww': '0.003468', '@mileycyrus': '0.002931', '@tommcfly': '0.002556', 'soooo': '0.002473', '@ddlovato': '0.002254', 'doesnt': '0.001800', '#followfriday': '0.001694', 'havent': '0.001559', '@jonasbrothers': '0.001553', 'isnt': '0.001336', '#fb': '0.001168', 'sooooo': '0.001041', 'awwww': '0.001037', 'tweetdeck': '0.000958', 'couldnt': '0.000939', ":'(": '0.000931', 'wasnt': '0.000913', '(via': '0.000896', '@davidarchie': '0.000892', '@donniewahlberg': '0.000865', '@jonathanrknight': '0.000825', '*sigh*': '0.000811', '@jordanknight': '0.000749', 'oooh': '0.000730', '@mitchelmusso': '0.000708', '(and': '0.000705', 'ohhh': '0.000693', 'ahhhh': '0.000664', '*hugs*': '0.000647', 'nooo': '0.000634', '#ff': '0.000628', '#squarespace': '0.000612', 'youre': '0.000609', 'p.s': '0.000594', 'noooo': '0.000588', 'b/c': '0.000581', 'ughh': '0.000575', 'goodmorning': '0.000555', 'mmmm': '0.000553', 're:': '0.000552', 'twitpic': '0.000540', 'soooooo': '0.000529', '@dougiemcfly': '0.000525', '@selenagomez': '0.000524', 'bgt': '0.000514', 'realised': '0.000508', "'em": '0.000503', 'thankyou': '0.000487', "ya'll": '0.000477', 'xxxx': '0.000471', 'booo': '0.000464', 'youu': '0.000458', '@dannymcfly': '0.000455', 'wouldnt': '0.000447', 'atleast': '0.000434', 'heyy': '0.000432', "'cause": '0.000432', 'ughhh': '0.000430', 'photo:': '0.000427', 'r.i.p': '0.000421', 'wooo': '0.000415', '@peterfacinelli': '0.000415', '@aplusk': '0.000409', 'tooo': '0.000408', 'tommorow': '0.000405', 'hayfever': '0.000405', 'a.m': '0.000401', '@joeymcintyre': '0.000399', 'goood': '0.000389', 'urgh': '0.000376', '@youngq': '0.000369', 'w/o': '0.000368', 'awsome': '0.000360', '(or': '0.000355', 'aswell': '0.000354', 'skool': '0.000354', 'tweetie': '0.000353', 'tomorow': '0.000346', 'boooo': '0.000336', '@shaundiviney': '0.000335', '#iranelection': '0.000335', ':-d': '0.000330', 'awwwww': '0.000330', '#seb-day': '0.000329', 'nooooo': '0.000327', 'yeahh': '0.000326', '@perezhilton': '0.000322', '@tomfelton': '0.000316', "g'night": '0.000313', 'twitterverse': '0.000311', '(y)': '0.000304', 'grrrr': '0.000299', '@officialtila': '0.000296', 'realise': '0.000289', '(not': '0.000286', '@kirstiealley': '0.000285'}

\end{lstlisting}

\end{document}